\documentclass[11pt,a4paper]{article}
\usepackage{latexsym,amssymb,amsmath,amsthm,amsfonts,enumerate,verbatim,xspace,
exscale}

\input xy
\xyoption{all} \CompileMatrices \UseComputerModernTips

\usepackage{graphicx,tabularx}
\usepackage{amsmath,amsbsy,amsfonts,amssymb}

\title{Gauge Transformations, Twisted Poisson Brackets and Hamiltonization of Nonholonomic Systems}

\author{{\sc{Paula Balseiro}\thanks{
         Universidade Federal Fluminense, Instituto de Matem\'atica, Rua Mario Santos Braga S/N, 24020-140, Niteroi, Rio de Janeiro, Brazil. \newline{\texttt{E-mail: pbalseiro@vm.uff.br}}}} , \ \
{\sc{Luis C. Garc\'ia-Naranjo}\thanks{
         Section de Math\'ematiques, Ecole Polytechnique F\'ed\'erale de Lausanne, CH-1015
Lausanne, Switzerland. \newline{\texttt{E-mail:
luis.garcianaranjo@epfl.ch}}}} }

\theoremstyle{plain}
\newtheorem{theorem}{Theorem}[section]

\newtheorem{proposition}[theorem]{Proposition}
\newtheorem{corollary}[theorem]{Corollary}
\newtheorem*{theorem*}{Theorem}
\newtheorem{remarkth}[theorem]{Remark}

\theoremstyle{definition}
\newtheorem{definition}{Definition}

\newenvironment{remark}{\begin{remarkth}\upshape}{\hfill$\diamond$\end{remarkth}}

\newenvironment{example}[1][Example]{\begin{trivlist}
\item[\hskip \labelsep {\bfseries #1}]}{\end{trivlist}}

\newenvironment{examples}[1][Examples]{\begin{trivlist}
\item[\hskip \labelsep {\bfseries #1}]}{\end{trivlist}}

\newcommand{\lcf}{\lbrack\! \lbrack}
\newcommand{\rcf}{\rbrack\! \rbrack}

\newcommand{\Xnh}{\mbox{$X_{\textup{nh}}$}}

\def\M{\mathcal{M}}
\def\C{\mathcal{C}}
\def\Ham{\mathcal{H}}
\def\Lag{\mathcal{L}}
\def\R{\mathbb{R}}
\def\D{\mathcal{D}}
\def\F{\mathcal{F}}
\def\RR{\mathcal{R}}
\def\L{\mbox{Leg}}

\def\RR{\mathcal{R}}
\def\vecOm{\boldsymbol{\Omega}}
\def\I{\mathbb{I}}

\def\vecom{\boldsymbol{\omega}}
\newcommand{\SO}{\mbox{$\textup{SO}$}}
\def\so{\mathfrak{so}}
\def\Lag{\mathcal{L}}
\def\se{\mathfrak{se}}
\def\vecep{\boldsymbol{\epsilon}}
\def\vecL{\boldsymbol{\lambda}}
\def\vecR{\boldsymbol{\rho}}
\def\vecgamma{\boldsymbol{\gamma}}

\begin{document}
\maketitle

\begin{abstract}
In this paper we study the problem of Hamiltonization of
nonholonomic systems from a geometric point of view. We use gauge
transformations by 2-forms (in the sense of \v{S}evera and Weinstein
\cite{SeveraWeinstein}) to construct different almost Poisson
structures describing the same nonholonomic system. In the presence
of  symmetries, we observe that these almost Poisson structures,
although gauge related, may have fundamentally different properties
after reduction, and that brackets that Hamiltonize the problem may
be found within this family. We illustrate this framework with the
example of rigid bodies with generalized rolling constraints,
including the Chaplygin sphere rolling problem. We also see how {\it
twisted Poisson brackets} appear naturally in nonholonomic mechanics
through these examples.

\end{abstract}

\tableofcontents

\section{Introduction} \label{S:Intro}

It is well known that the equations of motion for a mechanical
system with nonholonomic constraints do not arise from a variational
principle in the usual sense. As a consequence, they cannot be
formulated as a classical Hamiltonian system. Instead, they are
written with respect to an \emph{almost Poisson bracket} that fails
to satisfy the Jacobi identity. This formulation has its origins in
\cite{SchaftMaschke1994, Marle1998, IbLeMaMa1999} and others.

On the other hand, after a symmetry reduction, the equations of
motion of a number of examples allow a Hamiltonian formulation
(sometimes after a \emph{time reparametrization}) and one talks
about \emph{Hamiltonization}\footnote{A different meaning to
Hamiltonization is given in \cite{Mestdag-Variations} where the
authors study the \emph{unreduced} system in connection with the
inverse problem of the calculus of variations.}
(see \cite{Chapligyn_reducing_multiplier, EhlersKoiller, FedorovJovan,  Hoch, Fernandez, Naranjo2008,
JovaChap, Ohsawa} and others).

In this paper we employ recent developments of Poisson geometry to
study this phenomenon from a geometric perspective. We use
\emph{gauge transformations} by 2-forms as introduced by \v{S}evera
and Weinstein in \cite{SeveraWeinstein} to construct different
almost Poisson brackets describing the dynamics of the same
nonholonomic system. Although our interest is in almost Poisson
geometry, we consider more general objects known as \emph{almost
Dirac structures} \cite{Courant}, as they provide the most natural
setting  for the definition and study of gauge transformations.

We illustrate the need for our methods by working out the
Hamiltonization of the motion of rigid bodies that are  subject to
\emph{generalized rolling constraints}. These are nonholonomic
constraints that relate the angular velocity $\vecom$ of the body
and the linear velocity $\dot {\bf x}$ of its center of mass  in a
linear way (i.e. $\dot {\bf x}=A\vecom$ for a $3\times 3$ matrix
$A$). This type of constraints  contain the celebrated Chaplygin
sphere problem as a special case. The latter  concerns the motion of
an inhomogeneous sphere whose center of mass coincides with its
geometric center, that rolls without slipping on the plane. As a
consequence of the Hamiltonization, we are able to show complete
Liouville integrability of the reduced dynamics of any rigid body
subject to generalized rolling constraints.

Incidentally, during our discussion, we discover that \emph{twisted
Poisson brackets} \cite{SeveraWeinstein} appear in the study of
nonholonomic systems. In particular, we show that in the original
physical time (before the time reparametrization),
 the reduced dynamics of the Chaplygin sphere are formulated in terms of a twisted-Poisson bracket.
Although these structures do not in general satisfy the Jacobi
identity, they possess a fair amount of properties, including
foliations, that might imply an interesting interplay with dynamical
features. To date, the interest in these brackets has been mainly
geometrical.

\subsection{Hamiltonization}

Perhaps the most interesting example of Hamiltonization concerns the
Chaplygin sphere.
 Even though the formulation and integration of the equations of motion by Chaplygin dates back to 1903
  \cite{chapsphere}, the Hamiltonian structure of the reduced equations (after a time reparametrization) was only
  discovered in 2001 by Borisov and
Mamaev \cite{BorisovMamaev}.

Recently, Jovanovi{\'c} \cite{JovaChap} proved that the
multidimensional version of the Chaplygin sphere problem introduced
in \cite{FedorovKozlov} is also integrable and Hamiltonizable when
the vertical angular momentum is zero. This gives a  partial
solution to a problem that remained open for many years. His
approach to prove integrability involves in a crucial way the
Hamiltonization of the problem.
Another important example where the integration of a nonholonomic
system follows from its Hamiltonization is the multidimensional
Veselova system treated by Fedorov and Jovanovi\'c in
\cite{FedorovJovan}. We also mention the recent work of Ohsawa,
Fernandez,  Bloch and Zenkov \cite{Ohsawa} in connection with
Hamilton-Jacobi theory.

The relationship between Hamiltonization and integrability may have
been the original motivation for Chaplygin to consider the problem
of Hamiltonization back in 1911
\cite{Chapligyn_reducing_multiplier}. In this work, Chaplygin proved
the famous \emph{Chaplygin reducing multiplier Theorem} that applies
to the so-called \emph{$G$-Chaplygin systems}. These are
nonholonomic systems with the property that the tangent space to the
orbits of a symmetry group $G$ exactly complements the constraint
distribution on the tangent space $TQ$ of the configuration manifold
$Q$. Stated in  modern geometric terms, the Theorem says that if the
\emph{shape space} $Q/G$ is two-dimensional, and the reduced
equations have an invariant measure, then they can be put in
Hamiltonian form in the new time $\tau$ defined by  $d\tau
=\mbox{\small $\frac{1}{\varphi}$}\, dt$. The positive function $
\mbox{\small $\frac{1}{\varphi}$}:Q/G\to \R$ is known as the
 \emph{reducing multiplier}\footnote{We denote the reducing multiplier
by $\mbox{\small $\frac{1}{\varphi}$}$ instead of  $\varphi$ to be
consistent with our exposition which takes the Poisson rather than
the symplectic perspective.}. There is a very neat interpretation of
the multiplier $ \mbox{\small $\frac{1}{\varphi}$}$ in terms of the
invariant measure and as a conformal factor for an almost symplectic
form that describes the dynamics, see \cite{EhlersKoiller,
FedorovJovan, Hoch, Ohsawa}. This interpretation suggests that
geometric methods may be useful to understand Hamiltonization in
more general scenarios of nonholonomic systems with symmetry.

Recently, Fernandez, Mestdag and Bloch \cite{Fernandez}, derived a
set of coupled first order partial differential equations for the
multiplier  $ \mbox{\small $\frac{1}{\varphi}$}$ for $G$-Chaplygin
systems whose shape space has arbitrary dimension. Even more, the
set of equations found by the authors applies to general
nonholonomic systems with symmetry that are not necessarily
$G$-Chaplygin. This is  done by writing the reduced equations of
motion in Hamilton-Poincar\'e-D'Alembert form as described in
\cite{BKMM, Blochbook}. The issue of Hamiltonization is thus
reformulated as a problem of existence of a solution for the
aforementioned system of partial differential equations.

Our approach to Hamiltonization contains the same degree of
generality but is more intrinsic. Denote by $X_\RR$ the vector field
describing the dynamics on the reduced space $\RR$ and by $\Ham_\RR$
the reduced Hamiltonian. We formulate the reduced equations of
motion in almost Poisson form with respect to a \emph{collection} of
bivector fields $\pi_{\mbox{\tiny {red}}^B}$. Each member  in this
collection describes the reduced dynamics in (almost) Hamiltonian
form (i.e. $(\pi_{\mbox{\tiny {red}}^B})^\sharp(d\Ham_\RR)=-X_\RR$),
and arises as the reduction of a bivector field $\pi_{\mbox{\tiny {nh}}}^B$
associated to a bracket $\{\cdot , \cdot\}_{\mbox{\tiny {nh}}}^B$.
Such bracket  $\{\cdot , \cdot\}_{\mbox{\tiny {nh}}}^B$ is obtained through what we define
as a \emph{dynamical gauge transformation}  by a 2-form $B$ of the
noholonomic bracket $\{\cdot , \cdot \}_{\mbox{\tiny {nh}}}$ defined
in \cite{SchaftMaschke1994, Marle1998, IbLeMaMa1999} (see discussion in
subsection \ref{SS:intro-gauge}
below).

 Having said this, we reformulate
the issue of Hamiltonization  by requiring that one of the bivector
fields $\pi_{\mbox{\tiny {red}}^B}$ in the  collection described
above is \emph{conformally Poisson}, i.e.
\begin{equation}
\label{E:Conf-Poisson-Intro} [\varphi \, \pi_{\mbox{\tiny {red}}^B},
\varphi \, \pi_{\mbox{\tiny {red}}^B}] =0,
\end{equation}
for a positive function $\varphi$,\footnote{In fact  the function $\varphi:\RR\to \R^+$ is not arbitrary.
It is required  to be basic with respect to the fibered structure of the reduced space $\RR$.
Just as for $G$-Chaplygin systems, we think
of $\varphi: Q/G\to \R^+$.} and where $[\cdot , \cdot ]$ is
the Schouten bracket. The  scaling of $\pi_{\mbox{\tiny {red}}^B}$
by $\varphi$ is dynamically interpreted as the time
reparametrization $d\tau =\mbox{\small $\frac{1}{\varphi}$}\, dt$
(see Section \ref{Ss:ConformalFactor}).
Note that, for each  $\pi_{\mbox{\tiny {red}}^B}$, equation
\eqref{E:Conf-Poisson-Intro} locally defines a set of coupled first
order partial differential equations for $\varphi$. This seems to be
in  agreement with the results of Fernandez, Mestdag and Bloch
\cite{Fernandez}.

 If the symmetries are of $G$-Chaplygin type, the
bivector fields $ \pi_{\mbox{\tiny {red}}^B}$ are everywhere
non-degenerate and the equations of motion can be written with
respect to the associated almost symplectic form
$\Omega_{\mbox{\tiny {red}}^B}$. If a multiplier $\varphi$
satisfying \eqref{E:Conf-Poisson-Intro} exists, then
$\Omega_{\mbox{\tiny {red}}^B}$ is conformally closed,
($d(\mbox{\small $\frac{1}{\varphi}$}\Omega_{\mbox{\tiny
{red}}^B})=0$), and one speaks of \emph{Chaplygin Hamiltonization}
\cite{EhlersKoiller}.

The term \emph{Poissonization} was introduced by Fernandez, Mestdag
and Bloch in \cite{Fernandez} to refer to the case where the
bivector field $\pi_{\mbox{\tiny {red}}^B}$ satisfying
\eqref{E:Conf-Poisson-Intro} is degenerate. Their motivation to
distinguish this case is to study the relationship between
Hamiltonization and the existence of invariant measures for the
reduced equations.  To simplify  our exposition and to treat the
problem in a unified manner, we will not use their terminology and
simply talk about Hamiltonization whenever there exists a solution
to \eqref{E:Conf-Poisson-Intro}. We give a discussion on the
existence of invariant measures for nonholonomic systems admitting a
Hamiltonization in this generality in subsection
\ref{Ss:ConformalFactor}.

In general terms, the main contributions of this paper to the
problem of Hamiltonization are the clear geometric formulation of
the problem using recent developments on the field of Poisson
geometry (mainly those in \v{S}evera, Weinstein
\cite{SeveraWeinstein}) and the illustration of the usefulness of
these techniques in the study of rigid bodies subject to generalized
rolling constraints.

\subsection{Gauge transformations in nonholonomic mechanics}
\label{SS:intro-gauge}

The Hamiltonization of the Chaplygin sphere can be obtained in two
ways:

 In the first one,
described in Borisov and Mamaev \cite{BorisovMamaev2008}, one
performs the reduction in two stages.  In the first step one reduces
the translational symmetries of the problem corresponding to the
abelian action of $\R^2$. On the second  stage one reduces the
\emph{internal symmetry} corresponding to rotations about the
vertical axis described by the action of $\operatorname{S}^1$. The
second stage is performed  in Borisov and Mamaev
\cite{BorisovMamaev2008} using Routh's reduction method. The
Hamiltonian interpretation of the second reduction is delicate and
is studied in  Hochgerner, Garc\'ia-Naranjo \cite{Hoch}.
 Effectively, it is shown that in order to further reduce the system in a Marsden-Weinstein
fashion, the geometric data needs to be  \emph{modified} and the
authors propose a method for this called \emph{truncation}.

A second manner to achieve the Hamiltonization of the Chaplygin
sphere is described in Garc\'ia-Naranjo \cite{Naranjo2008}. In this
approach one first formulates the equations of motion in terms of an
\emph{affine almost Poisson bracket} that differs from the usual
nonholonomic bracket $\{\cdot , \cdot \}_{\mbox{\tiny nh}}$
considered in \cite{SchaftMaschke1994, Marle1998, IbLeMaMa1999} and
others. In this framework, the Hamiltonization can be obtained as a
single (one step) reduction by the symmetry group of rigid
transformations on the plane, $\operatorname{SE}(2)$.

In this paper we further elaborate on the latter approach. Our main
tool for this is to  incorporate \emph{gauge transformations} as
introduced in \v{S}evera and Weinstein \cite{SeveraWeinstein}.

Recall that the nonholonomic bracket $\{\cdot , \cdot
\}_{\mbox{\tiny nh}}$ considered in \cite{SchaftMaschke1994,
Marle1998, IbLeMaMa1999} is defined on functions on the constraint
phase space $\M$. The vector field $\Xnh$ on $\M$ that describes the
nonholonomic  dynamics is (almost) Hamiltonian with respect to this
bracket and with respect to the constrained Hamiltonian $\Ham_\M$.
Now, let $B$ be a 2-form on $\M$. If   $B$  satisfies a certain
technical condition, then a gauge transformation of the nonholonomic
bracket by $B$ defines a new bracket $\{\cdot , \cdot
\}^B_{\mbox{\tiny nh}}$ on $\M$ with the \emph{same (non-integrable)
characteristic distribution}.  As a consequence, the (almost)
Hamiltonian vector fields associated to  $\{\cdot , \cdot
\}^B_{\mbox{\tiny nh}}$  satisfy the nonholonomic constraints. If in
addition, the 2-form $B$ satisfies ${\bf i}_{\Xnh}B=0$, then we say
that $B$ defines a \emph{dynamical gauge transformation} and the
vector field $\Xnh$ is  (almost) Hamiltonian with respect to the
gauged bracket $\{\cdot , \cdot \}^B_{\mbox{\tiny nh}}$ and the
constrained Hamiltonian $\Ham_\M$.
%
%
%
%
In this way, we distinguish a family $\frak{F}$ of almost Poisson
structures that describe the dynamics of our nonholonomic system
corresponding to different dynamical gauge transformations. We  show
(Remark \ref{Rem:GeneralGauge}) that the affine almost Poisson
brackets defined in
 Garc\'ia-Naranjo \cite{Naranjo2008} are particular members in $\frak{F}$.

The main motivation to consider the large family $\frak{F}$ of
almost Poisson brackets for our nonholonomic system is to have a
larger choice of structures to describe the reduced dynamics and
hope to find one amongst them that Hamiltonizes the problem. In
particular, the Hamiltonization of the Chaplygin sphere arises as
the reduction of a dynamically gauged bracket $\{\cdot , \cdot
\}_{\mbox{\tiny nh}}^B$ that differs from the standard nonholonomic
bracket $\{\cdot , \cdot \}_{\mbox{\tiny nh}}$.

In fact, the members of the family $\frak{F}$ behave quite
differently after reduction. Our examples show that while some of
them might yield a true Poisson structure after reduction, others
yield an almost Poisson bracket with a non-integrable characteristic
distribution.

In order to study the gauge transformations of almost Poisson
brackets, we formulate the dynamics of nonholonomic systems on
almost Dirac structures \cite{Courant}. These are more general
geometric objects that provide the framework in which gauge
transformations are more natural. These structures had been already
considered in connection to nonholonomic mechanics by Yoshimura and
Marsden \cite{YoshimuraMarsdenI,YoshimuraMarsdenII}, and by Jotz
and Ratiu \cite{JotzRatiu}. However, the issue of Hamiltonization
and the incorporation of gauge transformations are not treated in
these works.

\subsection{Twisted Poisson brackets in nonholonomic mechanics}

It is well known that nonholonomic systems are formulated in terms
of  almost Poisson brackets that fail to satisfy the Jacobi
identity. However, very little research, if any, has been done in
understanding how far away these brackets are from being true
Poisson.

A very strong property of a Poisson manifold is that the
characteristic distribution is integrable and defines a foliation by
even dimensional leaves. This property is also shared by
conformally Poisson brackets whose bivector field $\pi$
satisfies $[\varphi \pi , \varphi \pi]=0$, for a certain positive
function $\varphi$ called the \emph{conformal factor}. These
brackets have been considered in the study of Hamiltonization of
nonholonomic systems, and, as mentioned before, the conformal factor
$\varphi$ defines the time reparametrization $d\tau =\mbox{\small
$\frac{1}{\varphi}$}\, dt$.

Another example of almost Poisson brackets that possess a foliation
by even dimensional leaves, is given by
 \emph{twisted Poisson brackets} that were introduced by Klim\v{c}\'\i k and
Str\"obl in \cite{KlimcikStrobl} and later in \v{S}evera and
Weinstein \cite{SeveraWeinstein} from a more geometric point of
view. Twisted Poisson brackets correspond to almost Poisson
structures whose associated bivector field $\pi$ satisfies
\begin{equation*}\frac{1}{2}
[\pi, \pi] = \pi^\sharp(\phi),
\end{equation*}
for a certain \emph{closed} 3-form $\phi$. To our knowledge, the
present paper is the first one to explore the connection between
this type of structures and nonholonomic mechanics.

In this paper we show that an almost Poisson structure with a
regular (constant rank), integrable characteristic distribution is
twisted (Corollary \ref{C:Int_Dist_implies_twist_Poisson}). As a
consequence we show (Remark \ref{R:Veselova-twisted}) that the
reduced equations of the classical Veselova problem
\cite{Veselova} can be formulated in terms of a twisted Poisson
bracket in the original physical time
 (prior to any time reparametrization).

We also show (Theorems \ref{T:ChaplyginIsTwisted} and
\ref{T:Rank1IsTwisted}) that the reduced equations of some examples
of rigid body motion with generalized rolling constraints, that
contain the Chaplygin sphere as a special case,  are described by a
twisted Poisson bracket in the original physical time. Moreover, we
give an explicit formula for the twisting closed 3-form $\phi$.

\subsection{Outline and main results of the paper}

The paper is organized as follows. In Section \ref{S:Examples} we
introduce our motivating examples, rigid bodies subject to
generalized rolling constraints. As mentioned before, these are
nonholonomic constraints that relate the linear velocity of the body
to its angular velocity via a $3\times 3$ matrix $A$. After writing
down the reduced equations of motion, we define two different
(almost) Poisson structures for the reduced equations according to
the rank  of $A$, that varies from $0$ to $3$. For each value of the
rank of $A$, we show that only one of the brackets is  Poisson
(conformally Poisson if rank $A=1,2$) while the other one possesses
a non-integrable characteristic distribution. The geometric
interpretation and construction of these brackets is one of the main
goals of the paper and is postponed to Section
\ref{S:Hamiltonization}, after the necessary tools are developed in
Sections \ref{S:Geometry} and \ref{S:Nonho}.

In Section \ref{S:Geometry} we develop the geometric background
needed for our purposes. We  focus on almost Dirac structures and
their gauge transformations, introduced respectively in
\cite{Courant} and \cite{SeveraWeinstein}, and we collect some new
results that are important in our study of nonholonomic systems.
Proposition \ref{L:Dirac2Section} gives a characterization
 of regular almost Dirac structures that is used in Corollary \ref{C:NonDegSection}
 to describe the structure of almost Poisson brackets having a regular characteristic distribution. Corollary
 \ref{C:Int_Dist_implies_twist_Poisson}
  shows that an almost Poisson bracket possessing a regular, integrable, characteristic
distribution is twisted.
In fact this result is proved in the more general setting of almost Dirac structures
in Theorem \ref{Prop:twisted}. We also mention
 Theorem \ref{T:GaugedRelation} that asserts
that any two regular almost Dirac structures defining the same
distribution are gauge related.


In Section \ref{S:Nonho} we make the connection between the
geometric methods developed in Section \ref{S:Geometry} and
nonholonomic mechanics. In particular, we construct the nonholonomic
bracket of \cite{SchaftMaschke1994, Marle1998, IbLeMaMa1999} using Corollary \ref{C:NonDegSection}
and the framework for nonholonomic mechanics described in Bates and Sniatycki \cite{BS93}. In Proposition
\ref{L:nhbracket} we show that the dynamics associated with this
bracket coincide with the formulation of nonholonomic mechanics on
almost Dirac structures considered in
\cite{YoshimuraMarsdenI,YoshimuraMarsdenII, JotzRatiu}. Next, we
define the notion of dynamical gauge transformations for a
nonholonomic system, and define a family $\mathfrak{F}$ of almost
Poisson brackets, possessing the same characteristic distribution,
and that describe our nonholonomic system. Finally, we discuss the
reduction of these brackets in the presence of symmetries and
introduce our working definition of Hamiltonization.

In Section \ref{S:Hamiltonization} we resume the study of rigid bodies subject to generalized
rolling constraints. In subsection \ref{SS:Geom_brackets} we show that the brackets given in
Section  \ref{S:Examples} to describe the reduced dynamics, arise as a reduction of different members of the
family $\mathfrak{F}$ (Theorems \ref{T:red_bracket} and \ref{T:red_bracket_gauge}).
In subsection \ref{SS:Hamiltonization-Integrability} we  establish the Hamiltonization of the reduced equations
in detail and we conclude their integrability. Finally, in subsection \ref{Ss:TwistedMechanics} we focus
on the twisted nature of the brackets that Hamiltonize the problem for the cases Rank $A=1,2$,
prior to the time reparametrization.


%
%
%
%

\section{Motivating Examples: Rigid bodies with Generalized Rolling Constraints}
 \label{S:Examples}

Consider the motion of a rigid body in space that evolves under
its own inertia and is subject to the
constraint that enforces the linear velocity of the center of mass,
$\dot {\bf x}$, to be linearly related to the angular velocity of
the body $\vecom$, i.e.,
\begin{equation}
\label{E:general_constraint} \dot {\bf x} =r A \vecom.
\end{equation}
 Both vectors $\dot {\bf x}$ and  $\vecom$  belong to $\R^3$ and are
written with respect to an inertial frame. The constant scalar $r$
has dimensions of length and is a natural length scale of the
system. The dimensionless constant $3\times 3$ matrix $A$ is given and satisfies
certain conditions that are made precise in the following Definition.

\begin{definition}
\label{condA}
The matrix $A$ is said to \emph{define a generalized
rolling constraint} if it satisfies one of the following conditions according to its rank:
\begin{itemize}
\item[$(i)$] $A= \left ( \begin{array}{cc} C & 0 \\ 0 & 1  \end{array} \right ),$ with $C\in \operatorname{SO}(2)$,
if $\mbox{rank} \, A =3$.
\item[$(ii)$] $A= \left ( \begin{array}{cc} C & 0 \\ 0 & 0  \end{array} \right ),$ with $C\in \operatorname{SO}(2)$,
if $\mbox{rank} \, A =2$.
\item[$(iii)$] $A={\bf e}_3 \, {\bf e}_3^T $,  if $\mbox{rank} \, A =1$, where
${\bf e}_3$ is the third canonical vector in $\R^3$, and $T$ denotes transpose.
 \item[$(iv)$] $A=0$ if $\mbox{rank} \, A =0$.
\end{itemize}

\end{definition}

The above conditions on $A$ can be relaxed (see Remark
\ref{rm:RelaxA} ahead). However, for simplicity, we will assume that
$A$ has the form given by one of the items of the above Definition.
If $A$ satisfies any of the the conditions of the above Definition
we say that \eqref{E:general_constraint} is a \emph{generalized
rolling constraint}.

Our terminology is motivated by a particular example:   the
Chaplygin sphere. The problem, introduced by Chaplygin in 1903
\cite{chapsphere}, concerns the motion of   a ball whose center of
mass coincides with its geometric center that rolls on the plane
without slipping. In this case, the matrix $A$ is given by
\begin{equation*}
A=\left ( \begin{array}{ccc} 0 & 1 & 0 \\ -1 & 0 & 0 \\ 0 & 0& 0
\end{array} \right ),
\end{equation*}
and $r$ is the radius of the sphere.

The motion of the Chaplygin ball has been the subject of much
research to our days. An important property is that (after a time
reparametrization) the reduced equations can be given a Hamiltonian
structure \cite{BorisovMamaev, Naranjo2008}. The geometry of the
Hamiltonization of the problem is intricate. In order to study this
phenomenon in a mathematically systematic fashion, we consider more
general possibilities for the matrix $A$.

The crucial property of $A$ that determines many of the dynamical and geometrical
features of the problem is its rank. The Chaplygin sphere
corresponds to the case  $\operatorname{rank} A=2$. Another familiar
case occurs when $\operatorname{rank} A=0$. In this case the
constraint \eqref{E:general_constraint} becomes  $\dot {\bf x}=0$
which can be interpreted as a conservation law for the free system
that states that the center of mass of the body is at rest in the
inertial frame. The motion of the system reduces to that of the
classical free rigid body.

We will also consider the cases where the rank of $A$ equals $3$ and
$1$ which, to our knowledge, have not yet been considered in the
literature.

\subsection{Generalities}

The configuration space for the system is $Q=\SO(3) \times \R^3$.
 Elements in $Q$  are of the form $q=(g,{\bf x})\in \operatorname{SO}(3)\times \R^3$. The vector ${\bf x} \in \R^3$
 is the position of the center of mass in space and the orthogonal matrix $g$ specifies the
orientation of the ball by relating two orthogonal frames, one
attached to the body and one that is fixed in space.
We will assume that the body frame has its origin at the center of mass and is aligned with the
principal axes of inertia of the body.
These frames
define the so-called \emph{space} and \emph{body coordinates}
respectively.

Recall that the Lie algebra $\so(3)$ can be identified with $\R^3$
equipped with the vector product via the \emph{hat map}:
\begin{equation}
\label{E:hat-map}
\boldsymbol{\eta}=(\eta_1,\eta_2,\eta_3)\mapsto \hat
{\boldsymbol{\eta}}= \left ( \begin{array}{ccc} 0& -\eta_3 & \eta_2
\\ \eta_3 & 0 & -\eta_1 \\ -\eta_2 &\eta_1 & 0
\end{array} \right ).
\end{equation}
Given a motion $(g(t), {\bf x}(t))\in Q $, the angular velocity
vector in space coordinates, $\boldsymbol{\omega}\in \R^3$, and the
angular velocity vector in body coordinates, $\vecOm \in \R^3$, are
respectively given by
\begin{equation*}
\hat {\boldsymbol{\omega}}(t)=\dot g(t) g^{-1}(t), \qquad \hat
\vecOm(t)= g^{-1}(t)\dot g(t),
\end{equation*}
and satisfy  $\vecOm= g^{-1} \vecom$. It will be useful to write the
constraint \eqref{E:general_constraint} in terms of the body angular
velocity as
\begin{equation}
\label{E:general_constraint-body} \dot {\bf x} = rAg \,\vecOm.
\end{equation}

The kinetic energy of the rigid body defines the Lagrangian
$\Lag:TQ\to \R$  by
\begin{equation}
\label{E:lag-rigidbody}
\Lag (g, \dot g, {\bf x}, \dot  {\bf x})=\frac{1}{2}( \I \vecOm)
\cdot  \vecOm  + \frac{m}{2}||\dot  {\bf x}||^2,
\end{equation}
where ``$\cdot$" denotes the Euclidean scalar product on $\R^3$, $m$
is the mass of the body and the $3\times 3$
 diagonal matrix
$\I$ is the inertia tensor with positive entries $I_1,I_2,I_3$.

\begin{remark}
\label{rm:RelaxA}
It is not hard to see that if the space axes are rotated by an element $h\in \operatorname{SO}(3)$, the
Lagrangian $\Lag$ is invariant and the constraint \eqref{E:general_constraint} is rewritten as
\begin{equation*}
\dot {\bf x} = rh^{-1}Ah \vecom.
\end{equation*}
Therefore,  the conditions for $A$ given in Definition \ref{condA}
can be relaxed by allowing conjugation by matrices    $h\in
\operatorname{SO}(3)$.
\end{remark}


\subsection{The equations of motion}

Let ${\bf p}=m\dot{\bf  x}$ be the  linear momentum of the body. In
accordance with the Lagrange-D'Alembert principle, the constraint
forces must annihilate any  velocity pair $(\dot {\bf x}, \vecOm)$
satisfying \eqref{E:general_constraint-body}. Therefore, the
equations of motion are given by
\begin{equation}
\label{E:Motion1} \dot {\bf p}= {\boldsymbol{\mu}}, \qquad \I \dot
\vecOm = \I \vecOm \times \vecOm -rg^{-1}A^T{\boldsymbol{\mu}},
\end{equation}
 where ``$\times$" denotes the vector product in $\R^3$ and
  the multiplier ${\boldsymbol{\mu}}\in \R^3$ is determined uniquely from the constraint \eqref{E:general_constraint-body}.

Differentiating \eqref{E:general_constraint-body} and using $\dot g
\vecOm=0$ we find ${\boldsymbol{\mu}}=mr Ag\dot \vecOm$. Thus, the
second equation in \eqref{E:Motion1} decouples from the first to
give
\begin{equation}
\label{E:Motion2}
 \I \dot \vecOm = \I \vecOm \times \vecOm -mr^2g^{-1}A^TA g \,\dot \vecOm.
\end{equation}
In principle, this equation should be complemented with the
\emph{reconstruction equation} $\dot g=g\hat \vecOm$. It will be shown
ahead that
  it suffices to consider the evolution of
the \emph{Poisson vector}  $\vecgamma:=g^{-1}{\bf e}_3$ that represents the vector ${\bf e}_3$
written in body coordinates. A direct calculation
gives
\begin{equation*}
\dot \vecgamma =\vecgamma \times \vecOm.
\end{equation*}
The decoupling in \eqref{E:Motion1} is due to the presence of
symmetries that will be discussed in detail in Section
\ref{SSS:reduction-example}. Once this equation is solved for $(g,
\vecOm)$, we  obtain ${\bf p}=mrAg\vecOm$ that follows from
\eqref{E:general_constraint-body}.

We introduce the kinetic momentum ${\bf K}\in \R^3$ by
\begin{equation}
\label{E:Kinetic_Momentum} {\bf K}:=\I \vecOm +mr^2g^{-1}A^TAg\,
\vecOm.
\end{equation}
This definition of the kinetic momentum allows us to define the
(reduced) Hamiltonian
\begin{equation}
\label{E:Ham} \Ham_\RR =\frac{1}{2} ( {\bf K} \cdot  \vecOm ),
\end{equation}
which coincides with the kinetic energy on the constraint space $\M$.

A direct calculation using \eqref{E:Motion2}  gives our final set of
equations
\begin{equation}
\label{E:Motion3}
\begin{split}
\dot {\bf K}= {\bf K} \times \vecOm, \qquad \dot \vecgamma
=\vecgamma \times \vecOm.
\end{split}
\end{equation}
To  understand why the above equations define a closed system for
$({\bf K}, \vecgamma)\in \R^3\times \R^3$, and to understand their
structure,  it is useful to perform a separate study for different
values of the rank of the matrix $A$. This will also show that the
Hamiltonian $\Ham_\RR$ can be considered as a function of ${\bf K}$
and $\vecgamma$.

\subsection{A pair of (almost) Poisson brackets for the equations of motion}

For each value of the rank of $A$ we will give two different
brackets that define the equations   of motion \eqref{E:Motion3}, with
respect to the reduced Hamiltonian $\Ham_\RR$. In general, these brackets
are almost Poisson, i.e. they do  not satisfy the Jacobi identity
 but we
will argue that one of them is more convenient than the other. They
will be denoted by $\{\cdot , \cdot \}_{\mbox{\tiny Rank{\it j}}}$
and $\{\cdot , \cdot \}'_{\mbox{\tiny Rank{\it j}}}$ where $j$
denotes the rank of the matrix $A$. Both  brackets
define the equations of motion \eqref{E:Motion3} in the sense that
the directional derivative of any function $f=f(\vecgamma, {\bf
K})\in C^\infty(\R^3\times \R^3)$ along the flow is given by $\dot
f=\{f,\Ham_\RR\}_{\mbox{\tiny Rank{\it
j}}}=\{f,\Ham_\RR\}'_{\mbox{\tiny Rank{\it j}}}$. The  geometric
interpretation of these brackets is the subject of the subsequent
sections. Concretely, in section \ref{S:Hamiltonization} (Theorems
\ref{T:red_bracket} and     \ref{T:red_bracket_gauge}), we will show
that the bracket $\{\cdot , \cdot \}_{\mbox{\tiny Rank{\it j}}}$
arises as the  reduction of the nonholonomic bracket introduced in
\cite{SchaftMaschke1994}, and that $\{\cdot , \cdot \}_{\mbox{\tiny
Rank{\it j}}}'$ arises as the reduction of a \emph{gauge
transformation} of the nonholonomic bracket.

The following
definitions will be useful in our discussion of the utility of the
two brackets:
\begin{definition} \label{D:Basics} Let $P$ be a  manifold equipped with an almost Poisson bracket $\{\cdot , \cdot \}$.
\begin{enumerate}
\item The \emph{(almost) Hamiltonian vector field} $X_f$ of a function $f\in C^\infty(P)$ is the vector field on $P$
defined as the usual derivation $X_f(g)=\{g,f\}$ for all $g\in
C^\infty(P)$.
\item  The \emph{characteristic distribution}
 of $\{\cdot , \cdot \}$ is the distribution on the manifold $P$ whose fibers are spanned by the (almost) Hamiltonian vector fields.
 \item Due to Leibniz condition of $\{\cdot , \cdot \}$, there is a bivector field $\pi \in \Gamma (\bigwedge ^2(TP))$
 such that for $f, g \in C^\infty(P)$ we have $\pi(df, dg) = \{f , g \}$.
 We say that $\pi$ is the \emph{bivector field associated to $\{ \cdot , \cdot \}$} and we
 denote by $\pi^\sharp : T^*P \rightarrow TP$ the map such that
  $\beta(\pi^\sharp(\alpha)) = \pi(\alpha, \beta)$. We will occasionally refer to  bivector fields simply as bivectors.
 Note that the characteristic distribution is the image of
 $\pi^\sharp$ and the Hamiltonian vector field  $X_f = - \pi^\sharp(df)$.
The 3-vector field $[\pi, \pi]$, where $[\cdot,\cdot]$ is the
Schouten bracket, may be different from zero, and it measures the
failure of the Jacobi identity through the relation
\begin{equation} \{f, \{g,h\}\}+ \{g,\{h,f\}\} + \{h,\{f,g\}\}=
\frac{1}{2} {\bf i}_{[\pi,\pi]}(df, dg, dh), \label{E:Jacobi}
\end{equation} for $f, g, h \in C^\infty(P)$.

 \item The bracket is called \emph{conformally Poisson} if there exists a strictly positive function $\varphi\in C^\infty (P)$ such that
 the bracket   $\varphi \{\cdot , \cdot \}$
satisfies the Jacobi identity, i.e. $[\varphi \pi, \varphi \pi ]=0$.
\end{enumerate}
\end{definition}

The well known symplectic stratification theorem states that the
characteristic distribution of a Poisson bracket is integrable and
its leaves are symplectic manifolds. Since multiplication of an
almost Poisson  bracket by a positive function  does not change the
characteristic distribution, a necessary condition for an almost
Poisson bracket to be conformally Poisson is that its characteristic
distribution be integrable.

We now come back to the discussion of our example for the different values of the rank of $A$.

\paragraph{If $A$ has rank 3.} In this case $A^{-1}=A^T$ and ${\bf K}=(\I +mr^2E)\vecOm$ where
$E$ denotes the $3\times 3$ identity matrix. It follows form
\eqref{E:Motion3} that the rotational motion of the body is the same
as that of a free rigid body whose total inertia tensor is $\I +
mr^2E$. It is trivial to write $\vecOm=(\I +mr^2E)^{-1}{\bf K}$ and
it is  clear that equations \eqref{E:Motion3} define a closed
system in $\R^3\times \R^3$.

The two brackets for the system for functions $f,g\in
C^\infty(\R^3\times \R^3)$ are given by
\begin{equation}
\label{E:Brackets-Rank3}
\begin{split}
&\{f,g\}_{\mbox{\tiny Rank3}}=-({\bf K}+mr^2\vecOm) \cdot \left  (
\frac{\partial f}{\partial {\bf K}} \times    \frac{\partial
g}{\partial {\bf K}} \right ) -\vecgamma \cdot \left (
\frac{\partial f}{\partial {\bf K}} \times    \frac{\partial
g}{\partial {\vecgamma}} -   \frac{\partial g}{\partial {\bf K}}
\times    \frac{\partial
f}{\partial {\vecgamma}} \right ), \\
&\{f,g\}'_{\mbox{\tiny Rank3}}=-{\bf K} \cdot \left ( \frac{\partial
f}{\partial {\bf K}} \times    \frac{\partial g}{\partial {\bf K}}
\right ) -\vecgamma \cdot \left ( \frac{\partial f}{\partial {\bf
K}} \times    \frac{\partial g}{\partial {\vecgamma}} -
\frac{\partial g}{\partial {\bf K}} \times    \frac{\partial
f}{\partial {\vecgamma}} \right ).
\end{split}
\end{equation}
The above brackets are quite different. On the one hand, the bracket
 $\{\cdot,\cdot\}'_{\mbox{\tiny Rank3}}$ satisfies the Jacobi identity. It in fact
 coincides with the Lie-Poisson bracket on the dual Lie algebra $\se(3)^*$.
On the other hand, the bracket  $\{\cdot,\cdot\}_{\mbox{\tiny Rank3}}$
is not even conformally Poisson as the following Proposition shows.
\begin{proposition}
\label{P:bracket-rank3} The characteristic distribution of the
almost Poisson bracket $\{\cdot,\cdot\}_{\mbox{\tiny \textup{Rank3}}}$  defined
in \eqref{E:Brackets-Rank3} is not integrable.
\end{proposition}
\begin{proof}
The (almost) Hamiltonian vector field $X_f$ of  a function $f\in
C^\infty (\R^3\times \R^3)$ corresponding to the bracket
$\{\cdot,\cdot\}_{\mbox{\tiny Rank3}}$ is given by
\begin{equation*}
X_f=\left  (({\bf K}+mr^2\vecOm) \times \frac{\partial f}{\partial
{\bf K}} + { \vecgamma}\times \frac{\partial f}{\partial
\vecgamma}\right ) \cdot  \frac{\partial }{\partial {\bf K}}
 + \left ( \vecgamma\times
\frac{\partial f}{\partial {\bf K}} \right ) \cdot  \frac{\partial
}{\partial \vecgamma}
\end{equation*}
and it is annihilated by the non-closed one-form
\begin{equation*}
\chi =\vecgamma \cdot d {\bf K} + ({\bf K}+mr^2 \vecOm)\cdot
d\vecgamma.
\end{equation*}

We have
\begin{equation*}
X_{K_1}=(K_3+mr^2\Omega_3)\frac{\partial}{\partial{K_2}} -(K_2+mr^2\Omega_2)\frac{\partial}{\partial{K_3}}
+\gamma_3\frac{\partial}{\partial \gamma_2}- \gamma_2\frac{\partial}{\partial \gamma_3}, \qquad
X_{\gamma_1}=\gamma_3\frac{\partial}{\partial{K_2}}-\gamma_2\frac{\partial}{\partial{K_3}},
\end{equation*}
and thus
\begin{equation*}
\chi ([X_{\gamma_1}, X_{K_1}])=- d\chi ( X_{\gamma_1},
X_{K_1})=-mr^2\left ( \frac{\gamma_3^2}{I_2+mr^2} +
\frac{\gamma_2^2}{I_3+mr^2} \right ) \neq 0.
\end{equation*}
This shows that the commutator $[X_{\gamma_1}, X_{K_1}]$ does not
belong to the characteristic distribution which is therefore not
integrable.
\end{proof}

Therefore, to obtain a true Hamiltonian formulation of the reduced
equations of motion in the case where the rank of $A$ is 3, one
needs to work with the bracket $\{\cdot,\cdot\}'_{\mbox{\tiny
Rank3}}.$

\paragraph{If $A$ has rank 2.} As mentioned before, this case has the Chaplygin sphere as a
particular example. The analysis of
the two brackets has been done in \cite{Naranjo2008}. We include it
here for completeness and to link it with clarity to other results
of the present work.

In view of the form of $A$ given in item ({\it ii}\,) of Definition
\ref{condA}, we can write $A^TA=E-{\bf e}_3{\bf e}_3^T$ and thus,
according to \eqref{E:Kinetic_Momentum}, we get
\begin{equation*}
{\bf K}=(\I+mr^2E)\vecOm - mr^2 (\vecOm \cdot \vecgamma) \vecgamma,
\end{equation*}
which is precisely the expression for the angular momentum about the
contact point for the Chaplygin sphere.

The angular velocity $\vecOm$ can be written in terms of ${\bf K}$
and $\vecgamma$ as
\begin{equation*}
\vecOm=(\I+mr^2E)^{-1}{\bf K}+mr^2\left ( \frac{{\bf K}\cdot
(\I+mr^2E)^{-1}\vecgamma}{||\vecgamma||^2-mr^2\vecgamma \cdot
(\I+mr^2E)^{-1}\vecgamma}  \right ) (\I +mr^2 E)^{-1}\vecgamma,
\end{equation*}
so both the equations \eqref{E:Motion3} and the Hamiltonian
$\Ham_\RR$ are well defined on $\R^3\times \R^3$.

In this case, the two brackets for the system for functions $f,g\in
C^\infty(\R^3\times \R^3)$ are given by
\begin{equation}
\label{E:Brackets-Rank2}
\begin{split}
&\{f,g\}_{\mbox{\tiny Rank2}}=-({\bf K}+mr^2\vecOm -mr^2
(\vecOm\cdot \vecgamma)\vecgamma ) \cdot \left  ( \frac{\partial
f}{\partial {\bf K}} \times    \frac{\partial g}{\partial {\bf K}}
\right ) -\vecgamma \cdot \left ( \frac{\partial f}{\partial {\bf
K}} \times \frac{\partial g}{\partial {\vecgamma}} - \frac{\partial
g}{\partial {\bf K}} \times    \frac{\partial
f}{\partial {\vecgamma}} \right ), \\
&\{f,g\}'_{\mbox{\tiny Rank2}}=-({\bf K}-mr^2 (\vecOm\cdot
\vecgamma)\vecgamma) \cdot \left  ( \frac{\partial f}{\partial {\bf
K}} \times    \frac{\partial g}{\partial {\bf K}} \right )
-\vecgamma \cdot \left ( \frac{\partial f}{\partial {\bf K}} \times
\frac{\partial g}{\partial {\vecgamma}} - \frac{\partial g}{\partial
{\bf K}} \times    \frac{\partial f}{\partial {\vecgamma}} \right ).
\end{split}
\end{equation}
None of the above brackets  satisfies the Jacobi identity but it is
preferable to consider $\{\cdot,\cdot\}'_{\mbox{\tiny
Rank2}}$. The reason is that this bracket is
conformally Poisson with conformal factor
\begin{equation}
\label{E:rank2_conf_factor}
\varphi(\vecgamma)=\sqrt{||\vecgamma||^2 - mr^2\, (\vecgamma \cdot
(\I+mr^2E)^{-1}\vecgamma)}.
\end{equation}
This important observation was first made in \cite{BorisovMamaev}.
The characteristic distribution of $\{\cdot,\cdot\}'_{\mbox{\tiny
Rank2}}$ is thus integrable. The generic leaves
are the level sets of the Casimir functions $C_1({\bf K}, \vecgamma)= {\bf K}\cdot \vecgamma$ and
$C_2(\vecgamma)=||\vecgamma||^2$. Another important feature of this bracket is that
it is \emph{twisted Poisson} (in the sense of \cite{KlimcikStrobl, SeveraWeinstein}) as will be shown in Section
\ref{Ss:TwistedMechanics} (Theorem \ref{T:ChaplyginIsTwisted}).

On the other hand, similar to Proposition \ref{P:bracket-rank3} we
have
\begin{proposition}[\cite{Naranjo2008}]
\label{P:bracket-rank2} The characteristic distribution of the
almost Poisson bracket $\{\cdot,\cdot\}_{\mbox{\tiny \textup{Rank2}}}$  defined
in \eqref{E:Brackets-Rank2} is not integrable.
\end{proposition}
This can be shown exactly as we did for Proposition
\ref{P:bracket-rank3}. Therefore if the rank of $A$ is 2, just as in
the case of rank 3,  a Hamiltonian formulation of the reduced
equations can only be obtained if we work with the bracket
$\{\cdot,\cdot\}'_{\mbox{\tiny Rank2}}$. However, in this case one
needs to multiply the bracket by a conformal factor. This can be
interpreted as a \emph{time reparametrization}, see the discussion
in Section \ref{Ss:ConformalFactor}.

\paragraph{If $A$ has rank 1.} Taking into account the form of $A$ given in item (\emph{iii})
of Definition \ref{condA} we have
 $A^TA={\bf e}_3{\bf e}_3^T$ and thus, in view
of \eqref{E:Kinetic_Momentum}, we get
\begin{equation*} \label{Ex:Kinet_Moment_Rank1}
{\bf K}=\I\vecOm + mr^2 (\vecOm \cdot \vecgamma) \vecgamma.
\end{equation*}
The expression for the angular velocity $\vecOm$  in terms of ${\bf
K}$ and $\vecgamma$ is
\begin{equation*} \label{E:Omega_in_terms_of_K_gamma}
 \vecOm=\I^{-1}{\bf K} -mr^2
\left ( \frac{{\bf K}\cdot
\I^{-1}{\vecgamma}}{||\vecgamma ||^2 +mr^2\, ({\vecgamma} \cdot \I^{-1}{\vecgamma} ) }
\right ) \I^{-1}{\vecgamma},
\end{equation*}
so again, both the equations \eqref{E:Motion3} and the Hamiltonian
$\Ham_\RR$  are well defined on $\R^3\times \R^3$.

This time, the two brackets for the system are given by
\begin{equation}
\label{E:Brackets-Rank1}
\begin{split}
&\{f,g\}_{\mbox{\tiny Rank1}}=-({\bf K} +mr^2 (\vecOm\cdot
\vecgamma)\vecgamma ) \cdot \left  ( \frac{\partial f}{\partial {\bf
K}} \times    \frac{\partial g}{\partial {\bf K}} \right )
-\vecgamma \cdot \left ( \frac{\partial f}{\partial {\bf K}} \times
\frac{\partial g}{\partial {\vecgamma}} -   \frac{\partial
g}{\partial {\bf K}} \times    \frac{\partial
f}{\partial {\vecgamma}} \right ), \\
&\{f,g\}'_{\mbox{\tiny Rank1}}=-({\bf K}-mr^2\vecOm +mr^2
(\vecOm\cdot \vecgamma)\vecgamma) \cdot \left  ( \frac{\partial
f}{\partial {\bf K}} \times    \frac{\partial g}{\partial {\bf K}}
\right ) -\vecgamma \cdot \left ( \frac{\partial f}{\partial {\bf
K}} \times    \frac{\partial g}{\partial {\vecgamma}} -
\frac{\partial g}{\partial {\bf K}} \times    \frac{\partial
f}{\partial {\vecgamma}} \right ),
\end{split}
\end{equation}
for functions $f,g\in C^\infty(\R^3\times \R^3)$.

The properties of the brackets above are very similar to those
obtained in the case where the rank of $A$
 is 2 except that the roles of $\{\cdot,\cdot\}_{\mbox{\tiny Rank1}}$ and $\{\cdot,\cdot\}'_{\mbox{\tiny Rank1}}$
are reversed.

This time one can show that
%


\begin{proposition} \label{P:ConfPoisson_Rank1}
The bracket $\{\cdot,\cdot
\}_{\mbox{\tiny \textup{Rank1}}}$ defined in \eqref{E:Brackets-Rank1} is conformally Poisson with conformal factor
\begin{equation*}
\varphi(\vecgamma)=\sqrt{||\vecgamma||^2 + mr^2\, (\vecgamma \cdot
\I^{-1}\vecgamma ) }.
\end{equation*}
\end{proposition}

\begin{proof} We have to prove that the scaled bracket on $\mathcal{R}$ defined as $\varphi\{\cdot ,
\cdot \}_{\mbox{\tiny Rank1}}$
 satisfies the Jacobi identity, i.e.,
\begin{equation*}
\varphi\{\varphi \{f_1 , f_2 \}_{\mbox{\tiny Rank1}} ,
f_3\}_{\mbox{\tiny Rank1}} + \varphi\{\varphi \{f_2 , f_3
\}_{\mbox{\tiny Rank1}} , f_1\}_{\mbox{\tiny Rank1}}+ \varphi
\{\varphi \{f_3 , f_1 \}_{\mbox{\tiny Rank1}} , f_2\}_{\mbox{\tiny
Rank1}}=0
\end{equation*}
for all $f_1,f_2,f_3\in C^\infty(\R^3\times \R^3)$. In view of the
derivation properties of the bracket, is enough to show the identity
for the  coordinate functions $K_i,\gamma_i$.
 In our
case, since $\{\gamma_i , \gamma_j \}_{\mbox{\tiny Rank1}}=0$, it is
immediate to check that the identity  holds  if two of the three
functions are $\gamma_i$'s. A long but straightforward
computation shows that the identity holds for the
following three choices of functions $f_1=K_1,  f_2=K_2,
f_3=\gamma_1$; $f_1=K_1, f_2=K_2, f_3=\gamma_3$ and $f_1=K_1,
f_2=K_2, f_3=K_3$. Since the definition of the bracket is symmetric
with respect to the coordinate functions $K_i,
\gamma_i$, and since the Jacobi identity trivially holds if two of
the three functions $f_1, f_2, f_3$ are equal, all of the other
cases are either trivial or analogous.
\end{proof}

Hence, the characteristic
distribution of $\{\cdot,\cdot \}_{\mbox{\tiny Rank1}}$ is integrable and the generic leaves are
again the level sets of the
Casimir functions $ C_1( {\bf K}, \vecgamma )={\bf K}\cdot \vecgamma$ and $C_2(\vecgamma)=||\vecgamma||^2$. It
will also be shown in Section \ref{Ss:TwistedMechanics} that $\{\cdot,\cdot\}_{\mbox{\tiny Rank1}}$ is twisted Poisson.

On the other hand, analogous to Propositions \ref{P:bracket-rank3} and
\ref{P:bracket-rank2} we have
\begin{proposition}
\label{P:bracket-rank1} The characteristic distribution of the
almost Poisson bracket $\{\cdot,\cdot\}'_{\mbox{\tiny \textup{Rank1}}}$  defined in \eqref{E:Brackets-Rank1} is
not integrable.
\end{proposition}
The proof is again similar.

Thus, this time the Hamiltonian structure of the reduced equations
can only be obtained with the bracket $\{\cdot,\cdot\}_{\mbox{\tiny
Rank1}}$, again through the multiplication by a conformal factor that
is interpreted as a time reparametrization.

\paragraph{If $A$ has rank 0.} In this case $A$ is the zero matrix and the constraints are holonomic and can be seen as a conservation
law for the standard free rigid body. We have ${\bf K}=\I\vecOm$ and
clearly the equations \eqref{E:Motion3} and the Hamiltonian
$\Ham_\RR$ are well defined  on $\R^3\times \R^3$. The two brackets
are given by
\begin{equation}
\label{E:Brackets-Rank0}
\begin{split}
&\{f,g\}_{\mbox{\tiny Rank0}}=-{\bf K} \cdot \left  ( \frac{\partial
f}{\partial {\bf K}} \times    \frac{\partial g}{\partial {\bf K}}
\right ) -\vecgamma \cdot \left ( \frac{\partial f}{\partial {\bf
K}} \times    \frac{\partial g}{\partial {\vecgamma}} -
\frac{\partial g}{\partial {\bf K}} \times    \frac{\partial
f}{\partial {\vecgamma}} \right ), \\
&\{f,g\}'_{\mbox{\tiny Rank0}}=-({\bf K}-mr^2\vecOm) \cdot \left  (
\frac{\partial f}{\partial {\bf K}} \times    \frac{\partial
g}{\partial {\bf K}} \right ) -\vecgamma \cdot \left (
\frac{\partial f}{\partial {\bf K}} \times \frac{\partial
g}{\partial {\vecgamma}} -   \frac{\partial g}{\partial {\bf K}}
\times    \frac{\partial f}{\partial {\vecgamma}} \right ).
\end{split}
\end{equation}
The situation is analogous to that of the case when the rank of $A$
is 3 but, once more, the roles of the brackets are reversed. While
$\{\cdot,\cdot \}_{\mbox{\tiny Rank0}}$ coincides with the Lie-Poisson
bracket in the dual Lie algebra $\se(3)^*$ (and hence satisfies the
Jacobi identity), we have
\begin{proposition}
\label{P:bracket-rank0} The characteristic distribution of the
almost Poisson bracket $\{\cdot,\cdot\}'_{\mbox{\tiny \textup{Rank0}}}$  defined in \eqref{E:Brackets-Rank0} is
not integrable.
\end{proposition}
The proof is identical to that of Proposition \ref{P:bracket-rank3}.

So in this case, the Hamiltonian structure of the reduced equations
\eqref{E:Motion3} can only be seen by working with the bracket
$\{\cdot,\cdot \}_{\mbox{\tiny Rank0}}$.

\subsection{Symmetries}
\label{SSS:reduction-example}

The reduced equations \eqref{E:Motion3} can be interpreted as the
output of a reduction process that we now explain. We begin by
noticing that the configuration space $Q=\operatorname{SO}(3)\times \R^3$ can be endowed with
 the Lie group  structure of the three dimensional euclidean
transformations $\operatorname{SE}(3)$. The group multiplication is
given by
\begin{equation*}
(g_1, {\bf x}_1)(g_2, {\bf x}_2)=(g_1g_2, g_1{\bf x}_2+{\bf x}_1).
\end{equation*}
Let $H$ be the Lie subgroup of $\operatorname{SE}(3)$ defined by
\begin{equation*}
H=\{(h,{\bf y})\in \operatorname{SE}(3) \, : \, h{\bf e}_3={\bf e}_3
\}.
\end{equation*}
For matrices  $A$ satisfying any of the conditions of Definition
\ref{condA}, it follows that $hA=Ah$ whenever $(h,{\bf y})\in H$.
 We consider the left action of $H$ on $Q$ by left multiplication.
The tangent lift of the action to $TQ$ maps
\begin{equation*}
(h,{\bf y}) : (g, {\bf x}, \vecom, \dot {\bf x})\mapsto (hg, h{\bf x}+{\bf y}, h \vecom, h\dot {\bf x}) \qquad
\mbox{ or} \qquad (h,{\bf y}) : (g, {\bf x}, \vecOm, \dot {\bf x})\mapsto (hg, h{\bf x}+{\bf y}, \vecOm, h\dot {\bf x}),
\end{equation*}
depending on the trivialization of $\operatorname{SO}(3)$ that one is working with.
Notice that the Lagrangian $\Lag$ given by \eqref{E:lag-rigidbody} is invariant under
the lifted action. Moreover, since $h$ commutes with $A$ for any $(h,{\bf y})\in H$,
the constraint \eqref{E:general_constraint} is also invariant.

The momenta $({\bf K}, {\bf p})$ are geometrically interpreted as coordinates on the fibers of
the (trivial) cotangent bundle $T^*Q$.  The constraint space $\M\subset T^*Q$ is determined
by the condition ${\bf p}=mrAg\vecOm$, so the triple $(g,{\bf x}, {\bf K})\in \operatorname{SO}(3)\times \R^3\times \R^3$
specifies a unique point in $\M$. Reciprocally, any point in $\M$ can be represented by
a triple $(g,{\bf x}, {\bf K})$.

By invariance of the Lagrangian and the constraints, the  lifted action of $H$ to
$T^*Q$ leaves $\M$ invariant and, therefore, restricts to $\M$. The restricted action is free
and proper so the orbit space $\RR:= \M/H$ is a smooth manifold. The reduced space $\RR$
can be identified with $\operatorname{S}^2\times \R^3$;  the projection $\rho:\M\to \RR$
is given by
\begin{equation}
\label{E:orbit-projection}
\rho(g,{\bf x}, {\bf K})=(\vecgamma, {\bf K}),
\end{equation}
where $\vecgamma =g^{-1}{\bf e}_3 \in \operatorname{S}^2$, and is a surjective submersion.
 The conditions $h{\bf e}_3={\bf e}_3$ and
$hA=Ah$, that are satisfied for $(h,{\bf y})\in H$, ensure that the above mapping is well defined
(in particular notice that ${\bf K}$ is invariant).
The reduced equations on $\RR$ are precisely \eqref{E:Motion3} when restricted to the level set
$||\vecgamma||=1$. Notice that $\RR$ inherits the (trivial) vector bundle structure $\operatorname{S}^2\times \R^3
\to \operatorname{S}^2$ from $\M$.

In this sense,  the entries of $\vecgamma$
should be considered as redundant coordinates for the sphere $\operatorname{S}^2$
and the entries of ${\bf K}$ as coordinates on the fibers of $\RR$.
Notice that, for any $j=0,\dots, 3$, both brackets
$\{\cdot , \cdot \}_{\mbox{\tiny Rank{\it j}}}$
and $\{\cdot , \cdot \}'_{\mbox{\tiny Rank{\it j}}}$ restrict to the level set $||\vecgamma||=1$
since $C_2(\vecgamma)=||\vecgamma||^2$ is a Casimir function.

\subsection{Kinematics and integrability of the constraint distribution}
\label{SS:Kinematics}

The constraint distribution on $Q$ defined by equation
\eqref{E:general_constraint} has fundamentally different properties
according to the rank of the matrix $A$ satisfying the conditions of
Definition \ref{condA}. On one extreme we have the case where $A=0$
and the distribution is integrable (the 3-dimensional integral
leaves are given by $\operatorname{SO}(3)\times \{ {\bf x} \} $ for
${\bf x}\in \R^3$). As mentioned before, in this case the
constraints are holonomic and the problem reduces to the classical
free rigid body problem (the center of mass of the body  ${\bf x}$
remains constant in our inertial frame).

The extreme opposite  case occurs when rank $A=3$. In this case the
corresponding distribution is \emph{completely nonholonomic} or
\emph{bracket-generating}, see e.g. \cite{Montgomery}. By Chow's
theorem, any  two points in the configuration space $Q$ can be
joined by a curve $(g(t),{\bf x}(t))$ satisfying the constraints.
Thus, at least at the kinematical level, there are no restrictions
on the values of ${\bf x}$.

The cases where the rank of $A$ is $1$ or $2$ lie in between the
situations described above. If the rank of $A=2$, the third
component $x_3$ of ${\bf x}$ remains constant during the motion.
This is in agreement with our observation that the Chaplygin sphere
problem is a particular case of this type of constraints - the
sphere rolls on a horizontal plane $x_3=const$. In this case, the
constraint distribution is non-integrable but is nevertheless
tangent to the foliation of $Q$ by $5$-dimensional leaves defined by
constant values of $x_3$.

Finally, for the case where the rank of $A$ equals $1$, the first
two components $x_1, \, x_2,$ of
 ${\bf x}$ remain constant during the motion. The body goes up or down along the $x_3$ axis
 at a speed that is proportional to its angular velocity about this axis.
 This time the constraint distribution is non-integrable but tangent to the $4$-dimensional
 leaves given by constant values of $x_1$ and $x_2$.

Without going into  technical definitions, we  simply  state that
the degree of non-integrability of the constraint distribution
increases with the rank of $A$, passing from an integrable
distribution if $A=0$ to a completely nonholonomic distribution if
rank $A=3$. It is interesting to see how this correlates with the
need of a gauge-transformation to Hamiltonize the problem (Remark
\ref{R:Gauge-vs-Integrability}).

%
%
%
%
%
%
%
%
%
%
%
%

\section{Geometric Setting}
\label{S:Geometry}
%

This section is concerned with the basic features of (almost) Dirac
structures \cite{Courant}, with focus on the regular case, as well
as their gauge transformations \cite{SeveraWeinstein}.  As we will see, these geometric structures
provide the setup that gives rise to the different brackets introduced in Section \ref{S:Examples}.
 Although we will be mostly interested in the geometry of bivector fields, our discussion is presented at the general level
 of (almost) Dirac structures, as they provide the
 framework in which gauge transformations are most natural.

\subsection{Dirac and almost Dirac structures}

A {\it Dirac structure} on a manifold $P$ is a subbundle $L$ of the Whitney sum $TP
\oplus T^*P$ such that
\begin{enumerate}
\item [$(i)$] $L$ is a maximal isotropic subbundle of $TP \oplus T^*P$
with respect to the pairing $\langle \cdot , \cdot \rangle$ given by
\begin{equation*} \langle(X,\alpha),(Y,\beta)\rangle = \alpha(Y)+\beta(X), \quad \mbox{for} \
(X,\alpha), (Y,\beta) \in TP \oplus T^*P.\end{equation*}

\item[$(ii)$] $\Gamma(L)$ is closed with respect to the Courant bracket
defined on $\Gamma(TP \oplus T^*P)$ given by
\begin{equation*} \lcf(X,\alpha),(Y, \beta)\rcf = ([X,Y], \pounds_X\beta - {\bf i}_Y
d\alpha),  \label{eq:CourantBracket} \end{equation*} for $(X,\alpha),
(Y,\beta) \in \Gamma(TP \oplus T^*P)$, i.e., $\lcf \Gamma(L),
\Gamma(L) \rcf \subseteq \Gamma(L)$.
\end{enumerate}
The underlying manifold  $P$ is sometimes referred to as a  \emph{Dirac manifold}.

Let $pr_1: TP \oplus T^*P \rightarrow TP$ be the projection onto the
first factor of  $TP \oplus T^*P$. A Dirac structure
$L$ on the manifold $P$ carries a Lie algebroid structure with
anchor $pr_1|_L: L \rightarrow TP$ and bracket given by the Courant
bracket $\lcf \cdot , \cdot \rcf$ restricted to $\Gamma(L)$. It
follows that, for a Dirac structure $L$, the distribution $pr_1(L)
\subset TP$ is integrable, i.e., $P$ can be decomposed into leaves
$\mathcal{O}$ such that for each $x \in P$, $T_x\mathcal{O} =
pr_1(L_x)$. If $pr_1(L)$ has constant rank (i.e., $pr_1(L_x) \subset
T_xP$ has the same dimension for all $x \in P$), then we say that
$L$ is a {\it regular} Dirac structure, and $pr_1(L)$ defines a
regular foliation. Just as a Poisson manifold $P$ is the disjoint
union of its symplectic leaves, each leaf of a Dirac manifold $P$
carries a presymplectic form.

\begin{examples} Let  $\Omega$ be a
closed 2-form and $\pi$ be a Poisson bivector field, and consider
the maps $\Omega^\flat: TP \rightarrow T^*P$ given by $\Omega^\flat
(X) = \Omega(\cdot, X)= - {\bf i}_X \Omega =$ for $X \in TP$ and
$\pi^\sharp:T^*P \to TP$ as in Definition \ref{D:Basics}. Then
$$L_\Omega := graph(\Omega^\flat) = \{(X,\alpha) \in TP \oplus T^*P \ : \ {\bf i}_X\Omega=-\alpha\}$$ and $$L_\pi := graph(\pi^\sharp) = \{(X,\alpha) \in TP \oplus T^*P \
: \ \pi^\sharp(\alpha) =X \}$$ are Dirac structures. Note that
$pr_1$ identifies $L_\Omega$ with $TP$ as Lie algebroids. Similarly,
$L_\pi$ can be naturally identified with $T^*P$, and the Lie
algebroid structure induced on $T^*P$ by $L_\pi$
 has anchor $\pi^\sharp:T^*P \rightarrow TP$, and bracket
\begin{equation*}
[\alpha, \beta]_\pi = \pounds_{\pi^\sharp(\alpha)} (\beta)-
\pounds_{\pi^\sharp(\beta)}(\alpha) - d(\pi(\alpha, \beta))  =
\pounds_{\pi^\sharp(\alpha)} \beta -  {\bf i}_{\pi^\sharp(\beta)}
d\alpha. \label{eq:LAPoisson}
\end{equation*}
This bracket is uniquely characterized by $[df, dg]_\pi = d\{f,g\}$ and the Leibniz
identity.
\end{examples}

\medskip

If a subbundle $L$ of $TP \oplus T^*P$ satisfies ({\it i}\,) above,
but not necessarily ({\it ii}\,), then $L$ is called an {\it almost
Dirac structure}. Condition ({\it ii}\,) is called the {\it
integrability condition}. We say that $L$ is a {\it regular} almost
Dirac structure when the distribution $pr_1(L)$ on $P$ has constant
rank. Notice that this distribution might not be integrable in the
almost Dirac case.

\begin{examples} If $\Omega$ is an arbitrary 2-form or $\pi$ an
arbitrary bivector field, then their graphs $L_\Omega$ and $L_\pi$
are almost Dirac structures. The failure of the integrability with
respect to the Courant bracket of $L_\Omega$ and $L_\pi$ is measured
by $d\Omega$ and $\frac{1}{2}[\pi,\pi]$, respectively. For $L_\pi
=graph (\pi^\sharp)$, the distribution $\pi^\sharp
(T^*P)=pr_1(L_\pi)$ is generally non-integrable. If it has constant
rank we call the almost Poisson structure \textit{regular}. Note
that the bracket $[\cdot, \cdot ]_\pi$ defined as in
(\ref{eq:LAPoisson}) is $\R-$bilinear, skew-symmetric and satisfies
the Leibniz identity. However, in general, $\pi^\sharp$ does not
necessarily preserve the bracket; instead, (see e.g. \cite{BCrainic}),
\begin{equation}\pi^\sharp ( [\alpha,\beta]_\pi) =
[\pi^\sharp(\alpha),\pi^\sharp(\beta)] - \frac{1}{2} {\bf
i}_{\alpha\wedge \beta} [\pi, \pi], \qquad \mbox{for } \alpha, \beta
\in \Omega^1(P). \label{E:not_Morph}
\end{equation}
\end{examples}

Note that an almost Dirac structure $L$ on $P$ is of the form $L_\pi
= graph(\pi^\sharp)$ for a bivector $\pi$ if and only if
\begin{equation}
TP \cap L = \{0\}, \label{Eq:TPcapL}
\end{equation}
and $L$ is of the form
$L_\Omega=graph(\Omega^\flat)$ for a 2-form $\Omega$ if and only if
$T^*P \cap L = \{0\}$, see \cite{Courant}.
Another example of an almost Dirac structure that will be very
useful for our purposes is given by $L \subset TP \oplus
T^*P$ defined as
\begin{equation}
L:=\{(X, \alpha) \in TP \oplus T^*P \ : \ X \in F , \ {\bf i}_X
\Omega |_F = -\alpha|_F \}, \label{Ex:almostDirac}
\end{equation}
where $F \subset TP$ is a subbundle, $\Omega$ is a 2-form on $P$ and
$\cdot \, |_F$ denotes the point-wise restriction to $F$.  If the
subbundle $F$ is an integrable distribution and $\Omega$ is closed,
then $L$ is a Dirac structure.


\begin{proposition} \label{L:Dirac2Section}
 The following statements hold:

\begin{enumerate} \item [(i)] There is a one-to-one correspondence between regular almost Dirac
structures $L \subset TP \oplus T^*P$ and pairs $(F, \Omega_F)$,
where $F$ is a regular distribution on $P$ and $\Omega_F \in \Gamma(\bigwedge^2F^*)$.
\item [(ii)] Let $F \subset TP$ be a regular
distribution on P. Given a section $\Omega_F \in
\Gamma(\bigwedge^2F^*)$, there exists a 2-form $\Omega$ on $P$ such
that $\Omega |_{F} = \Omega_F$.
\end{enumerate}
\end{proposition}

\begin{proof}
($i$) Let $L \subset TP \oplus T^*P$ be a regular almost Dirac
structure
 with distribution $F:=pr_1(L)
\subset TP$ on $P$ (not necessarily integrable). Consider the
section $\Omega_F$ in $\Gamma(\bigwedge ^2 F^*)$ given, at each $x
\in P$, by
$$\Omega_F(x)(X_{(x)}, Y_{(x)})= - \alpha_{(x)} (Y_{(x)}), \qquad
\mbox{for} \ X,Y \in \Gamma(F) \mbox{ such that } (X_{(x)},
\alpha_{(x)}) \in L_{x}.
$$
It is a straightforward computation to see that $\Omega_F$ is well
defined, i.e., it is independent of the choice of $\alpha$.
Conversely, given a regular distribution $F$ on $P$ and $\Omega_F
\in \Gamma(\bigwedge^2F^*)$, we may define the subbundle $L \subset
TP \oplus T^*P$ as the pairs $(X,\alpha)$ such that $X\in F$ and
${\bf i}_X \Omega_F = -\alpha|_F$. 

($ii$) Let ${\mathcal W} \subset TP$ be a regular smooth
distribution such that it is a complement of $F$ on $P$, i.e.,
$T_xP=F_x \oplus
{\mathcal W}_x$ for each $x \in P$ 
(e.g., ${\mathcal W}_x$ can be chosen to be the orthogonal
complement of $F_x$ with respect to a Riemmanian metric). The
2-form $\Omega$ on $P$ can be defined by
$$ \Omega(X,Y)= \Omega_F(X_F,Y_F),$$   for $X,Y \in
\mathfrak{X}(P)$ such that $ X=X_F + X_{\mathcal W}$ \, and \, $
Y=Y_F + Y_{\mathcal W}$, where $X_F, Y_F \in \Gamma(F)$ and
$X_{\mathcal W}, Y_{\mathcal W} \in \Gamma({\mathcal W})$.
Differentiability of $\Omega$  follows from its definition and the
smoothness of $F$ and $\mathcal{W}$.

\end{proof}


\begin{corollary} \label{C:DiracForm} Given a regular almost Dirac structure
$L$, there exists a 2-form $\Omega$ on $P$ and a regular
distribution $F\subseteq TP$ such that $L$ is written in the form
\eqref{Ex:almostDirac}.
\end{corollary}

\begin{remark}
 Note that the 2-form $\Omega \in \Omega^2(P)$ is not uniquely defined and, in general, there is no canonical choice for
 it.
\end{remark}

Given a subbundle  $F\subseteq TP$, we say that a section $\Omega_F$
in $\Gamma(\bigwedge^2F^*)$ is \textit{nondegenerate} if it is
nondegenerate as a bilinear form on $F$ at each point. It follows
from \eqref{Ex:almostDirac} that $\mathrm{Ker}(\Omega_F)= L\cap
TP$, and as a consequence of condition \eqref{Eq:TPcapL} we obtain

\begin{corollary} \label{C:NonDegSection} Let $L$ be a regular  almost Dirac
structure and $(F, \Omega_F)$ the pair associated to it in the sense
of Proposition \ref{L:Dirac2Section}. Then $\Omega_F$ is
nondegenerate if and only if $L$ is the graph of a bivector field
$\pi$. Explicitly, the relation between $(F,\Omega_F)$ and $\pi$ is
\begin{equation*} \pi^\sharp(\alpha) = - X \quad
\mbox{if and only if} \quad {\bf i}_{X} \Omega_F= \alpha|_F,
\label{E:FormBracket}
\end{equation*}
where $X \in  \Gamma(F)$ and $\alpha \in \Omega^1(P)$.
\end{corollary}

Following notation of Definition \ref{D:Basics}, if $\{ \cdot ,
\cdot \}$ is the bracket associated to the bivector field $\pi$ in
the above Corollary, then $\{ f ,g\} =  \Omega_F(X_f, X_g)$, for all
$f,g \in C^\infty (P).$

\subsection{Twisted Poisson and  twisted Dirac structures}
\label{Sec:Twisted}


Poisson structures may be viewed as encoding integrability in two
levels: first, the characteristic distribution $\pi^\sharp(T^*P)
\subseteq TP$ is integrable, i.e., tangent to leaves; second each
leaf carries a nondegenerate 2-form that is closed (and this leads
to the Jacobi identity). Twisted Poisson structures are special
types of almost Poisson structures that retain the integrability of
$\pi^\sharp(T^*P)$ but allow the leafwise 2-form to be non closed.
These objects turn out to be related  to Hamiltonization. We
start with the more general notion of twisted Dirac structures.


Consider a closed 3-form $\phi$ on $P$, and define the $\phi$-{\it
twisted Courant bracket} \cite{SeveraWeinstein} as follows:
\begin{equation} \lcf(X,\alpha),(Y, \beta)\rcf_\phi = ([X,Y], \pounds_X\beta - {\bf i}_Y
d\alpha + {\bf i}_{X\wedge Y} \phi ), \label{Eq:twistedCourant}
\end{equation} for $(X,\alpha)$ and $(Y, \beta)$ in $\Gamma(TP \oplus
T^*P)$. Now, a subbundle $L$ of $TP \oplus T^*P$ is a $\phi$-{\it
twisted Dirac structure} \cite{SeveraWeinstein} if $L$ is maximal isotropic with respect to
$\langle\cdot, \cdot \rangle$ and the integrability condition $$\lcf
\Gamma(L), \Gamma(L) \rcf_\phi \subseteq \Gamma(L)$$ is satisfied.

As in the ordinary case, a twisted Dirac structure $L$ on $P$
induces a Lie algebroid on $L$ given by the anchor map $pr_1 |_L$
and the bracket $\lcf \cdot, \cdot \rcf_\phi |_{\Gamma(L)}$.
Therefore $pr_1(L)$ is an integrable distribution on $P$.

\begin{examples}  If $\Omega$ is any 2-form on $P$, then
$L_\Omega=graph(\Omega^\flat)$ is $(d\Omega)$-twisted Dirac. One may check
that a bivector field $\pi$ on $P$ such that  $L_\pi=graph(\pi^\sharp)$
is $\phi$-twisted Dirac verifies (see \cite{SeveraWeinstein})
\begin{equation}\frac{1}{2} [\pi, \pi] = \pi^\sharp(\phi).
\label{E:twisted_Schout}
\end{equation}
\end{examples}

The following result gives more examples:

\begin{theorem} \label{Prop:twisted} Let $L \subset TP \oplus T^*P$  be a regular almost Dirac structure
such that $pr_1(L) \subset TP$ is an integrable distribution
on $P$. Then,  there exists an exact 3-form $\phi$ with respect to which $L$ is a $\phi$-twisted Dirac structure.
\end{theorem}

\begin{proof}
Let $F := pr_1(L) \subset TP $ and $\Omega_F$ be the section in
$\Gamma(\bigwedge^2F^*)$ associated to $L$ given by Proposition
\ref{L:Dirac2Section}. Since $F$ is integrable, $\Omega_F$ defines a
2-form $\Omega_{\mathcal O}$ on each leaf $\mathcal{O}$ where $F_x =
T_x\mathcal{O}$ at each $x\in P$. By Corollary \ref{C:DiracForm} there exists a  2-form $\Omega$ on $P$
such that $\iota^*_{\mathcal O}
\Omega = \Omega_{\mathcal O}$ where $\iota_{\mathcal O} : {\mathcal
O} \hookrightarrow P$ is the inclusion. We assert that $L$ is a $(d
\Omega)$-twisted Dirac structure. In fact, for $(X,\alpha)$ and $(Y,
\beta)$ in $\Gamma(L), $
$$ \lcf (X,\alpha) ,(Y,\beta)\rcf_{(d\Omega)}= ([X, Y], \pounds_X \beta- {\bf
i}_Y d\alpha + {\bf i}_{X\wedge Y} d \Omega)  \in \Gamma(L)$$ if and
only if $$ {\bf i}_{[X, Y]} \Omega|_F = -(\pounds_X \beta- {\bf i}_Y
d\alpha + {\bf i}_{X\wedge Y} d \Omega)|_F.$$ Since
 $F$ is an integrable distribution we obtain that,
\begin{eqnarray*}-(\pounds_X \beta- {\bf i}_Y d\alpha + {\bf i}_{X\wedge Y} d \Omega)|_F
&=& -\pounds_X (\beta|_F)+ {\bf i}_Y d(\alpha|_F) - {\bf i}_{X\wedge
Y} d (\Omega|_F) \\ &=& \pounds_X {\bf i}_Y \Omega_F - {\bf i}_Y
d{\bf i}_X \Omega_F - {\bf i}_{X\wedge Y} d
 \Omega_F = {\bf i}_{[X,Y]} \Omega_F.
\end{eqnarray*} which completes the proof.
\end{proof}

\begin{remark} \label{R:nonCanChoice} \
\begin{enumerate} \item [$(i)$] Note that if $L$ is a $\phi-$twisted Dirac structure
then $L$ is also twisted with respect to any closed 3-form $\phi'$
such that $(\phi-\phi')$ vanishes on the leaves.
\item [$(ii)$]  There is no canonical choice for the 3-form $\phi$ given in
Theorem \ref{Prop:twisted}. \end{enumerate}
\end{remark}

\medskip

\subsubsection*{Twisted Poisson bivectors.}

Bivector fields $\pi$ such that $L_\pi = graph(\pi^\sharp)$  is a $\phi$-twisted Dirac structure are called $\phi$-{\it twisted Poisson bivectors} \cite{KlimcikStrobl, SeveraWeinstein}, i.e., $\pi$  verifies condition (\ref{E:twisted_Schout}).  We are especially interested in these kind of
bivector fields since, as we will see, they appear naturally in
the examples of nonholonomic systems introduced in Section \ref{S:Examples}. If $\{ \cdot, \cdot \}$ is the
bracket given by the $\phi$-twisted Poisson structure $\pi$, then
relation (\ref{E:Jacobi}) becomes
\begin{equation*} \{f, \{g,h\}\}+ \{g,\{h,f\}\} + \{h,\{f,g\}\} +
\phi(X_f, X_g, X_h)=0, \label{twistedJAC}
\end{equation*} for $f,g, h \in C^\infty(P)$ and $X_f = \{ \cdot , f\}$.
So the failure of the Jacobi identity is controlled by the closed
3-form $\phi$.

\begin{example} If $\pi^\sharp:T^*P\to TP$ is an isomorphism, then $\pi$ is
$\phi$-twisted and $\phi=d\Omega$ where $\Omega$ is the
nondegenerate 2-form associated with $\pi$, i.e., $\Omega^\flat
\circ \pi^\sharp = \textup{Id}$ where as usual $\Omega^\flat(X) =
-{\bf i}_X\Omega$. Other, less trivial examples, will be presented ahead in Corollary \ref{C:Int_Dist_implies_twist_Poisson}.
\end{example}

\medskip

Let $\pi$ be a $\phi$-twisted Poisson structure on the manifold $P$
and $[\cdot, \cdot]_\pi$ be the bracket on $T^*P$ given by
(\ref{eq:LAPoisson}). Note that $\pi^\sharp$ does not preserve this
bracket. However, using (\ref{E:not_Morph}) and
(\ref{E:twisted_Schout}) we obtain
$$[\pi^\sharp(\alpha),\pi^\sharp(\beta)] = \pi^\sharp \left( [\alpha,\beta]_\pi + {\bf i}_{\pi^\sharp(\alpha)\wedge \pi^\sharp(\beta)} \phi \right), $$
for 1-forms $\alpha, \beta$ on $P$.  The $\phi$-twisted Courant
bracket induces a modification of the bracket (\ref{eq:LAPoisson})
via the identification of $T^*P$ and $L_\pi$, \begin{equation*}
[\alpha, \beta]_\phi = \pounds_{\pi^\sharp(\alpha)} \beta - {\bf
i}_{\pi^\sharp(\beta)} d\alpha + {\bf i}_{\pi^\sharp(\alpha)\wedge
\pi^\sharp(\beta)} \phi, \label{eq:LAtwisted}
\end{equation*}
such that $(T^*P, [\cdot, \cdot]_\phi, \pi^\sharp)$ is
a Lie algebroid 
\cite{SeveraWeinstein} (see also \cite{BCrainic}).
 The characteristic distribution $\pi^\sharp(T^*P)$ defines an integrable distribution on $P$ (that may be singular).
 Each leaf  $\mathcal{O}$ of the corresponding foliation of $P$
 is endowed with a non-degenerate 2-form $\Omega_{\mathcal O}$ that is not necessarily closed. If $\pi$ is  $\phi$-twisted,
 then $d\Omega_{\mathcal O} = \iota^*_{\mathcal O} \phi$, where
$\iota_{\mathcal O}: \mathcal{O} \hookrightarrow P$ is the inclusion.

Important examples of twisted Poisson structures are contained in the following Corollary of
Theorem \ref{Prop:twisted}:
\begin{corollary}
\label{C:Int_Dist_implies_twist_Poisson} Let $\pi$ be a bivector
field on $P$ with an integrable regular characteristic distribution.
Then,  there exists an exact 3-form $\phi$ on $P$ with respect to
which $\pi$ is  $\phi$-twisted.
\end{corollary}

\begin{remark} \label{R:Veselova-twisted} {\it Mechanical Example.}
 It is shown in \cite{Naranjo2007} that the (semi-direct)
product reduction of the Veselova system yields a regular
conformally Poisson bracket on the reduced space. It follows that
its characteristic distribution is integrable and thus, by Corollary
\ref{C:Int_Dist_implies_twist_Poisson},
 it is  also twisted-Poisson.
This is a first example of a nonholonomic system whose reduced
equations are formulated in terms of a twisted-Poisson bracket.
Other examples (related to the motion of a rigid body with
generalized rolling constraints) are made explicit in Section
\ref{Ss:TwistedMechanics}.
\end{remark}

\begin{remark}
An interesting question, that remains to be answered, is to give a characterization of
almost Poisson brackets possessing an integrable characteristic distribution
that is non-regular.
\end{remark}

\subsubsection*{Regular conformally Poisson bivectors.}
An interesting class of almost Poisson structures admitting an
integrable characteristic distribution is given by conformally
Poisson structures. Recall from Section \ref{S:Examples} that they
are bivector fields $\pi$ for which exists a strictly positive
function $\varphi \in C^\infty(P)$, such that $\varphi \pi$ is
Poisson. A conformally Poisson manifold $(P, \pi)$ is the disjoint
union of conformally symplectic leaves.

Note that this property is stronger than asking for $(P, \pi)$ to be
a Jacobi manifold since a conformal factor implies the global
existence of a function such that $\varphi \pi$ is Poisson, while in
Jacobi manifolds the factor $\varphi$ may be only locally defined.



From Theorem \ref{Prop:twisted} we observe that any regular bivector
admitting a conformal factor is also a twisted Poisson bivector. The
following Proposition explains the relation between these two
properties.

\begin{proposition} \label{Prop:ConformallyTwisted}
Let $\pi$ be a regular conformally Poisson bivector field on $P$
with conformal factor $\varphi$. Let $\Omega \in \Omega^2(P)$ be as
in Corollary \ref{C:DiracForm}. Then any closed 3-form verifying
\eqref{E:twisted_Schout} for $\pi$ coincides with
$\displaystyle{\frac{1}{\varphi} d\varphi \wedge \Omega}$ on the
leaves.
\end{proposition}

\begin{proof} Since $\pi$ admits a conformal factor $\varphi \in C^\infty(P)$,
then $[\pi, \pi] = \frac{2}{\varphi} X_\varphi \wedge \pi.$ On the
other hand, if $\Omega$ is the 2-form on $P$ associated to $\pi$
given by Corollary \ref{C:DiracForm}, then for $g_1, g_2 \in
C^\infty(P)$ we have $\Omega(X_{g_1}, X_{g_2}) = \pi(dg_1,dg_2)$.
Thus, for $g_1, g_2, g_3 \in C^\infty(P)$ we obtain
\begin{equation*} \frac{1}{2}[\pi, \pi](dg_1, dg_2 ,
dg_3) =  \frac{1}{\varphi} X_\varphi \wedge \pi (dg_1 , dg_2 , dg_3)
= -\frac{1}{\varphi} d\varphi \wedge \Omega (X_{g_1}, X_{g_2},
X_{g_3}).
\end{equation*} Then we conclude that any closed 3-form $\phi$ satisfying
\eqref{E:twisted_Schout} coincides with $(\frac{1}{\varphi} d\varphi
\wedge \Omega)$ on the leaves.

\end{proof}

Let $(P,\pi)$ be a regular conformally Poisson manifold.
A  2-form $\Omega$ on $P$ satisfying the conditions of Corollary \ref{C:DiracForm}
 verifies that $\iota^*_{\mathcal O} \Omega$ is conformally symplectic on each leaf
$\mathcal{O}$. However, $\Omega$ may not necessarily be conformally
closed.



\subsection{Gauge transformations.}

In this section we will consider a natural action of the abelian group of
2-forms on $P$ on the almost Dirac structures on $P$. This action is
given by {\it gauge transformations of  almost Dirac
structures by 2-forms}, and it was introduced in
\cite{SeveraWeinstein}.

More precisely, consider an almost Dirac structure $L$ in $TP \oplus
T^*P$. A {\it gauge transformation by the 2-form $B$}
is a map $\tau_B : L \rightarrow TP \oplus T^*P$,  given by
$\tau_B((X,\alpha)) = (X,\alpha + {\bf i}_XB )$ for $(X, \alpha) \in
L$. The subbundle $\tau_B(L)$ of $TP \oplus T^*P$ given by
\begin{equation*}
\label{E:Def-gauge-transf}
\tau_B(L) = \{(X,\alpha
+ {\bf i}_XB ) \ : \ (X,\alpha) \in L \}
\end{equation*}
 is an almost Dirac
structure. If the 2-form $B$ is closed and $L$ is Dirac then
$\tau_B(L)$ is again Dirac. Thus, the 3-form $dB$ is what determines
 the integrability with respect to the Courant bracket. It
is a direct computation to see that if $L$ is a $\phi$-twisted Dirac
structure, then the gauge transformation of $L$ by the
2-form $B$ is $(\phi-dB)$-twisted Dirac (see e.g.
\cite{SeveraWeinstein}).

If $L_1$ and $L_2$ are almost Dirac structures on $P$ and there
exist a 2-form $B$ on $P$ such that $\tau_B(L_1) = L_2$, then we say
that $L_1$ and $L_2$ are {\it gauge equivalent} or {\it gauge
related}.

Note that a gauge transformation does not modify the distribution
$pr_1(L)$. So, for a Dirac structure $L$, the foliation associated
to $L$ will be the same as the one associated to $\tau_B(L)$.
However, the presymplectic form on each leaf is modified by the
pullback of $B$ to the leaf. If $L$ is a regular almost Dirac
structure determined by the pair $(F, \Omega_F)$ in the sense of
item ($i$) of Proposition \ref{L:Dirac2Section}, a gauge
transformation by the 2-form $B$ corresponds to the operation:
 \begin{equation} \tau_B: (F, \Omega_F) \to
(F, \Omega_F -B|_F). \label{Eq:GaugeSection} \end{equation}

\begin{theorem} \label{T:GaugedRelation}
Any two regular almost Dirac structures $L_1$ and $L_2$ are gauge
related if and only if $pr_1(L_1)= pr_1(L_2)$.
\end{theorem}

\begin{proof} It remains to prove the ``only if'' part of the statement.
Let us denote $F:=pr_1(L_1)= pr_1(L_2)$ and let $\Omega^1_F$ and
$\Omega^2_F$ be the 2-sections associated to $L_1$ and $L_2$
respectively (Proposition \ref{L:Dirac2Section} ($i$)). Define the
section $B_F \in \Gamma(\bigwedge^2(F^*))$  by $B_F:= \Omega^1_F -
\Omega^2_F$ and let $B \in \Omega^2(P)$ such that $B|_F= B_F$ (Prop.
\ref{L:Dirac2Section} ($ii$)). We claim that $\tau_B L_1 = L_2$. In
fact, if $(F,\Omega_F^B)$ is the pair associated to the almost Dirac
structure $\tau_BL_1$, then by equation \eqref{Eq:GaugeSection},
$$\Omega^B_F = \Omega^1_F - B_F = \Omega^2_F.
$$ Since the sections associated to $\tau_BL_1$ and $L_2$ coincide, by Proposition \ref{L:Dirac2Section} ($i$) we conclude that $\tau_B L_1 =
L_2$ which means that $L_1$ and $L_2$ are gauge related.

\end{proof}

We are especially interested in gauge transformations of almost
Poisson structures.  Consider the almost Poisson manifold $(P, \pi)$
and a 2-form $B$ on $P$. Then, the gauge transformation of $L_{\pi} :=
graph(\pi^\sharp)$ is
 $\tau_B(L_{\pi}) = \{(X,\alpha + {\bf i}_XB ) \in TP \oplus
T^*P \ : \ X=\pi^\sharp(\alpha) \}$ which does not necessarily
correspond to the graph of a new bivector $\pi^B$. A necessary and
sufficient condition for this to happen is that
\begin{equation*} \tau_B(L_{\pi}) \cap TP = \{0\} \label{E:condGaugedbivector}, \end{equation*} which is equivalent to the
fact that the endomorphism $\textup{Id} + B^\flat\circ
\pi^\sharp:T^*P \rightarrow T^*P$ is invertible
\cite{SeveraWeinstein}. Indeed, if such a bivector field $\pi^B$ exists, then, in view of \eqref{E:Def-gauge-transf},
 for any 1-form $\alpha$ on $P$ we have
$$
\tau_B \left( (\pi^\sharp(\alpha) \, , \,  \alpha) \right) =
\left (\pi^\sharp(\alpha) \, , \, \alpha + {\bf i}_{\pi^\sharp(\alpha)} B \right )
= \left ((\pi^B)^\sharp \left ( \alpha + {\bf i}_{\pi^\sharp(\alpha)} B \right ) \, , \, \alpha + {\bf
i}_{\pi^\sharp(\alpha)} B \right).
$$
Thus, $\pi^B$ is characterized by the condition
\begin{equation}
\label{E:condition-gauge}
(\pi^B)^\sharp \left ( \alpha + {\bf i}_{\pi^\sharp(\alpha)} B \right ) =
\pi^\sharp(\alpha),
\end{equation}
and we can write
\begin{equation} (\pi^{B})^\sharp = \pi^\sharp \circ (\textup{Id}- B^\flat\circ
\pi^\sharp )^{-1}. \label{Eq:GaugedBivector} \end{equation}

In particular, if $\pi$ and $\pi^B$ are non-degenerate bivector
fields, equation (\ref{Eq:GaugedBivector}) is equivalent to
$$((\pi^B)^\sharp )^{-1} = (\pi^\sharp)^{-1} - B^\flat.$$ Therefore, any two 2-forms on a manifold are gauge
related. This is not necessarily the case  with bivector fields.
A necessary condition is that their characteristic distributions coincide. In view of Theorem
\ref{T:GaugedRelation}, this condition is also sufficient if such distributions are regular.


\medskip

Although gauge related bivectors have the same characteristic
distribution, their  Schouten brackets
$[\pi,\pi]$ and $[\pi^B, \pi^B]$ may not coincide.
The following Proposition  makes this precise.
\begin{proposition}[\cite{SeveraWeinstein}]
\label{P:gauge-twist}
 If a $\phi$-twisted Poisson
bivector $\pi$ is gauge related with another bivector $\pi^B$ via
the 2-form $B$, then $\pi^B$ is $(\phi-dB)$-twisted. That is,
\begin{equation*} \frac{1}{2}[\pi^B, \pi^B] = (\pi^B)^\sharp(\phi-dB).
\label{Eq:SchoutenPhi}
\end{equation*}
In particular, if $\pi$ is Poisson, then $\pi^B$ is $(-dB)$-twisted.
\end{proposition}
%

\bigskip




\section{Applications to Nonholonomic Systems and Hamiltonization}
\label{S:Nonho}

In what follows, we will analyze the geometry of nonholonomic
systems in the framework presented in the previous section. We
introduce gauge transformations of the bracket describing the
nonholonomic dynamics in order to study the process of
Hamiltonization.

\subsection{Nonholonomic systems}
\label{S:Nonholonomic_Systems}

A nonholonomic system consists of an $n$-dimensional configuration manifold $Q$ with
local coordinates ${\bf q}\in U\subset \R^n$, a Lagrangian
$\Lag:TQ\rightarrow \R$ of the form $\Lag (\dot {\bf q}, {\bf q}) = \frac{1}{2}
\mathcal{G}({\bf q})(\dot {\bf q},\dot {\bf q}) - V({\bf q})$,  where $\mathcal{G}$ is a kinetic energy
metric on $Q$ and $V:Q\rightarrow \R$ is a potential, and
 a regular non-integrable distribution $\D\subset TQ$  that describes
the kinematic nonholonomic constraints. In coordinates, the
distribution $\D$ is defined by the
equation
\begin{equation}
\label{E:Constraints_Lag_form} {\vecep}({\bf q})\, \dot {\bf q}=0,
\end{equation}
where $\vecep({\bf q})$ is a $k\times n$ matrix of constant rank $k$ where $k<n$ is the number of constraints.
The entries of  $\vecep({\bf q})$ are the components of the
$\R^k$-valued \emph{constraint 1-form} on $Q$, $
\vecep:=\vecep({\bf q})\; d{\bf q}. $

The dynamics of the system are governed by the Lagrange-D'Alembert
principle. This principle states that the forces of constraint
annihilate any virtual displacement, so they perform no work during
the motion. The equations of motion take the form
\begin{equation}
\label{E:Eqns_of_motion_Lag_form} \frac{d}{dt} \left (\frac{\partial
\Lag}{\partial \dot {\bf q}}\right ) -\frac{\partial \Lag}{\partial
{\bf q}}={\boldsymbol{\mu}}^T\,\vecep({\bf q})  .
\end{equation}
Here ${\boldsymbol{\mu}}:TQ\to \R^k$ is an  $\R^k$-valued function whose entries are referred to as
Lagrange multipliers. Under our assumptions, it is uniquely
determined by the condition that the constraints
(\ref{E:Constraints_Lag_form}) are satisfied. The equations
(\ref{E:Eqns_of_motion_Lag_form}) together with the constraints
(\ref{E:Constraints_Lag_form})  define a vector field
$\mbox{$Y_{\mbox{\tiny nh}}^{\D}$}$ on $\D$  whose integral curves
describe the motion of the nonholonomic system. A short calculation
shows that along the flow of $\mbox{$Y_{\mbox{\tiny nh}}^{\D}$}$, the
energy function $E_\Lag:= \frac{\partial \Lag}{\partial \dot {\bf q}} \cdot
\dot {\bf q} - \Lag$, is conserved.

The above equations  of motion can be written as a first order
system on the cotangent bundle $T^*Q$ via the standard Legendre
transform, $\L:TQ\rightarrow T^*Q$, that defines canonical
coordinates $({\bf q},{\bf p})$ on $T^*Q$ by the rule $\L:({\bf q},\dot
{\bf q})\mapsto ({\bf q},{\bf p}=\partial \Lag / \partial \dot {\bf q})$. The
Legendre transform is a global diffeomorphism by our assumption that
$\mathcal{G}$ is a metric.

The Hamiltonian function, $\Ham:T^*Q\rightarrow \R$, is defined  in
the usual way $\Ham:=E_\Lag \circ \L^{-1}$. The equations of motion
(\ref{E:Eqns_of_motion_Lag_form}) are shown to be equivalent to
\begin{equation}
\label{E:Eqns_of_motion_Ham_form} \dot {\bf q} = \frac{\partial
\Ham}{\partial {\bf p}}, \qquad \dot {\bf p}= - \frac{\partial
\Ham}{\partial {\bf q}} + {\boldsymbol{\mu}}^T \vecep({\bf q}),
\end{equation}
and the constraint equations (\ref{E:Constraints_Lag_form}) become
\begin{equation}
\label{E:Constraints_Ham_form} \vecep({\bf q}) \frac{\partial
\Ham}{\partial {\bf p}}=0.
\end{equation}
The above equation defines the \emph{constraint submanifold}
$\M=\L(\D) \subset T^*Q$.  Since the Legendre transform is linear on
the fibers, $\M$ is a vector sub-bundle of $T^*Q$ that for each
$q\in Q$ specifies an $n-k$ vector subspace of $T^*_qQ$.

Equations (\ref{E:Eqns_of_motion_Ham_form}) together with
(\ref{E:Constraints_Ham_form}) define the vector field
$\mbox{$X_{\mbox{\tiny nh}}$}$ on $\M$, that  describes the motion
of our nonholonomic system in the Hamiltonian side and is the push
forward of the vector field $\mbox{$Y_{\mbox{\tiny nh}}^{\D}$}$ by the
Legendre transform. The  vector field
$\mbox{$X_{\mbox{\tiny nh}}$}$  is defined uniquely in an intrinsic
way by the equation
\begin{equation}
\label{E:Eqns_of_motion_Ham_form_intrinsic} {\bf
i}_{\mbox{$X_{\mbox{\tiny nh}}$}} \,\iota^* \Omega_Q =
\iota^*(d\Ham +\boldsymbol{\mu}^T  \tau^* \vecep),
\end{equation}
where $\Omega_Q$ is the canonical symplectic form on $T^*Q$,
$\iota:\M \hookrightarrow T^*Q$ is the inclusion and $\tau:
T^*Q\rightarrow Q$ is the canonical projection. The constraints
(\ref{E:Constraints_Ham_form}) and their derivatives are
intrinsically written as
\begin{equation}
\mbox{$X_{\mbox{\tiny nh}}$}\in \C :=T\M\cap \F, \label{E:C}
\end{equation}
where $\F$ is the distribution on $T^*Q$ defined as $\F:=\{ v\in
T(T^*Q): \langle \tau^*\vecep ,  v\rangle =0\}$. Denote by $\Omega_\M$ the pull-back of $\Omega_Q$ to  $\M$, i.e. $\Omega_\M:=
\iota^* \Omega_Q$.
%
The following Proposition is of great importance for our setup of the equations
of motion as an almost Hamiltonian system.
\begin{proposition}[\cite{Weber1986, BS93}]
The distribution $\C$ on $\M$ defined by \eqref{E:C} is regular, non-integrable, and the point-wise restriction of
$\Omega_\M$ to $\C$, denoted by $\Omega_\C$, is non-degenerate.
\end{proposition}
The non-integrability of $\C$ is a direct consequence of the non-integrability of $\D$. One shows that the rank of
$\C$ is $2(n-k)$ and that
 along $\M$
we have the symplectic decomposition
\begin{equation}
\label{E:symp-decomp}
T_\M(T^*Q)=\C \oplus \C^{\Omega_Q},
\end{equation}
where $ \C^{\Omega_Q}$ denotes the symplectic orthogonal complement to $\C$.

Since $\tau^*\vecep$ vanishes on $\C$, by restricting
(\ref{E:Eqns_of_motion_Ham_form_intrinsic}) to $\C$ and denoting
$\Ham_\M := \iota^*\Ham \in C^\infty(\M)$, the equations of motion
can be written in the appealing format
\begin{equation}
\label{E:Eqns_of_motion_Almost_Almost_Ham}
\mathbf{i}_{\mbox{$X_{\mbox{\tiny nh}}$}} \Omega_\C= (d\Ham_\M)_\C,
\end{equation}
where $(d\Ham_\M)_\C$ is the point-wise restriction of $d\Ham_\M$ to
$\C$. The above equation uniquely defines the vector field
$\mbox{$X_{\mbox{\tiny nh}}$}$ and is central in our treatment;
with this in mind,
 we collect the data of the nonholonomic system in the triple $(\M, \Omega_\C ,   \Ham_\M)$. 

Even though (\ref{E:Eqns_of_motion_Almost_Almost_Ham}) defines the
vector field $\mbox{$X_{\mbox{\tiny nh}}$}$ uniquely, and resembles
a  classical Hamiltonian system,  notice that since  the distribution $\C$ is
non-integrable, then $\Omega_\C$ is a section in $\bigwedge^2\C^* \rightarrow \M$ (not a 2-form).

Let $(\M, \Omega_\C ,  \Ham_\M)$ be a nonholonomic system. For every
$f \in C^\infty (\M)$, let $\mbox{$X_{f}$}$ denote the unique vector
field on $\M$ with values in $\C$ defined by the equation
\begin{equation} \label{E:Almost_Almost_Ham}
\mathbf{i}_{\mbox{$X_{f}$}} \Omega_\C= (d f)_\C,
\end{equation} where $(df)_\C$ denotes the
point-wise restriction of $df$ to $\C$. The vector field $X_f$ defined by
equation (\ref{E:Almost_Almost_Ham}) is called the {\it the (almost)
Hamiltonian vector field associated to $f$}.

Since $\Omega_\C$ is nondegenerate, by Corollary
\ref{C:NonDegSection} there is a unique bivector field
$\pi_{\mbox{\tiny nh}}$ on $\M$ associated to the pair
$(\C,\Omega_\C)$, that is $\pi_{\mbox{\tiny nh}}^\sharp(\alpha) =
-X$ if and only if $\mathbf{i}_{\mbox{$X$}} \Omega_\C= \alpha |_\C.$
On exact forms we have, for $f\in C^\infty(\M)$,
\begin{equation}
\label{E:DefBracket}\mathbf{i}_{\mbox{$X_{f}$}} \Omega_\C= (d f)_\C
\quad \mbox{if and only if}  \quad \pi_{\mbox{\tiny nh}}^\sharp(df)
= - X_f,
\end{equation} which is consistent with notation of Definition
\ref{D:Basics}. The bracket $\{\cdot, \cdot\}_{\mbox{\tiny nh}}$ on
functions on $\M$ associated to the bivector $\pi_{\mbox{\tiny nh}}$
  {\it describes the dynamics} in the sense that
\begin{equation*} \label{E:Eq_motion_nhBracket}
\mbox{$X_{\mbox{\tiny nh}}$} (f)(m) = \mbox{$X_{\Ham_\M}$}(f)
(m) = \{f,\Ham_\M \}_{\mbox{\tiny nh}}(m) \qquad \mbox{for all} \ f \in
C^\infty(\M).\end{equation*}
It follows from \eqref{E:DefBracket} that the characteristic distribution
of the bracket $\{\cdot, \cdot\}_{\mbox{\tiny nh}}$ is $\C$. Since $\C$ is non-integrable
then $\{\cdot, \cdot\}_{\mbox{\tiny nh}}$ is an almost Poisson bracket that
does not satisfy the Jacobi identity.

\begin{remark}
If the constraint distribution $\D$ were integrable, the same would be true for
the distribution $\C$. Let $\mathcal{N}\subset \M$ be a leaf of the corresponding (regular) foliation of $\M$
(i.e. $\C_x = T_x\mathcal{N}$ for all $x\in \mathcal{N}$). In view of \eqref{E:symp-decomp} the submanifold
$\mathcal{N}$ is symplectic. Therefore, in this case, our construction of $\{\cdot, \cdot\}_{\mbox{\tiny nh}}$ coincides with the
usual construction of the Dirac bracket on each leaf $\mathcal{N}$ of the foliation of $\M$ (see \cite{Dirac, Courant}). Hence, in this case,
the Jacobi identity holds.
\end{remark}



Inspired by equation \eqref{Ex:almostDirac} and our data, it is natural to
define the almost Dirac structure $L_{\mbox{\tiny \textup{nh}}}$ on $\M$ by
\begin{equation} L_{\mbox{\tiny \textup{nh}}}:= \{ (X,\alpha) \in T\M \oplus T^*\M \ : \ X \in
\C, \ {\bf i}_X \Omega_\M |_\C = -\alpha |_\C \}, \label{E:DiracNH}
\end{equation}
as considered (up to a minus sign) in \cite{YoshimuraMarsdenI, YoshimuraMarsdenII}, and subsequently in
\cite{JotzRatiu}.
%


\begin{proposition}\label{L:nhbracket}
Let $\pi_{\mbox{\tiny \textup{nh}}}$ be the bivector field defined
in \eqref{E:DefBracket} and $\{\cdot, \cdot\}_{\mbox{\tiny
\textup{nh}}}$ the corresponding bracket. The following statements hold:
\begin{enumerate}
\item [$(i)$] The almost Dirac structures $L_{\pi_{\mbox{\tiny \textup{nh}}}}:=
 graph(\pi_{\mbox{\tiny \textup{nh}}}^\sharp)$ 
and $L_{\mbox{\tiny \textup{nh}}}$ given in \eqref{E:DiracNH} coincide.
\item [$(ii)$] The almost Poisson bracket $\{\cdot,
\cdot\}_{\mbox{\tiny \textup{nh}}}$ coincides with the classical
almost Poisson  bracket for nonholonomic systems defined in \textup{\cite{SchaftMaschke1994, Marle1998}}.
\end{enumerate}
\end{proposition}

\begin{proof}
({\it i}\,) It is immediate since both almost Dirac structure are
defined by the same pair $(\C, \Omega_\C)$.

%
%

({\it ii}\,) The  bracket on $\M$ for nonholonomic systems introduced in
\textup{\cite{SchaftMaschke1994, Marle1998}} was shown in \cite{Cantrijn99}
to be given by
\begin{equation}
\label{E:almost_Poisson}
\{ f, g \} = \Omega_\M (\mathcal{P}
\mbox{$Y_{\bar f}$}, \mathcal{P}\mbox{$Y_{\bar g}$}), \qquad
\mbox{for } f, g \in C^\infty(\M),
\end{equation}
where $\mathcal{P}:T_\M(T^*Q)\to \C$ is the projector associated to the symplectic decomposition
\eqref{E:symp-decomp}, and $\mbox{$Y_{\bar f}$}$ the free Hamiltonian vector
field on the symplectic manifold $(T^*Q, \Omega_Q)$ defined by
$\mathbf{i}_{\mbox{$Y_{\bar f}$}} \Omega_Q= d\bar f $, where $\bar f$ is an arbitrary smooth extension
of $f$ to $T^*Q$.

It is easy to check that along $\M$ one has ${\bf i}_{\mbox{$\mathcal{P} Y_{\bar f}$}}\Omega_\C =(df)_\C$,
so $\mathcal{P}Y_{\bar f}$ coincides with the almost
Hamiltonian vector field $X_f$ defined by equation
\eqref{E:Almost_Almost_Ham}. Therefore, for
any $f,g \in C^\infty(\M)$,
$$\{f,g\} = \Omega_\M (\mathcal{P} \mbox{$Y_{\bar f}$},
\mathcal{P}\mbox{$Y_{\bar g}$}) =\Omega_\C (X_f, X_g)= -
df(\pi^\sharp_{\mbox{\tiny nh}}(dg) ) = \{f,g \}_{\mbox{\tiny nh}}.
$$
\end{proof}

The second item of the above Proposition should not be surprising since the
 expression \eqref{E:almost_Poisson} is the nonholonomic version of the Dirac bracket
 (see discussion in \cite{IbLeMaMa1999, Cantrijn99}). Hence, its description naturally
 falls in the ambit of almost Dirac structures as described above.
As a consequence of the above Proposition, we can equivalently describe
 our nonholonomic system with the triple $(\M,
\pi_{\mbox{\tiny nh}}, \Ham_\M)$.

\begin{definition} \label{D:NHdescription}
Let $(\M, \pi_{\mbox{\tiny nh}},\Ham_\M)$ be a
nonholonomic system.
\begin{enumerate}\item
The bivector field $\pi_{\mbox{\tiny nh}}$ on $\M$ given by
(\ref{E:DefBracket})  is called {\it the nonholonomic bivector
field} and the bracket $\{\cdot, \cdot\}_{\mbox{\tiny nh}}$ is
called the {\it nonholonomic bracket}. \item  We say that an almost
Dirac structure $L$ {\it describes the dynamics of the nonholonomic
system} if the pair $(-X_{\mbox{\tiny nh}}, d \Ham_\M) \in
\Gamma(L)$.
\end{enumerate}
\end{definition}


%
%



\medskip

\subsection{Gauge transformations of the nonholonomic bracket}
\label{SS:GaugeTransf}

The main idea of using gauge transformations in our setting is that it opens  the
possibility to modify the geometric structure on $\M$ that describes the dynamics.

Consider the nonholonomic system $(\M, \pi_{\mbox{\tiny nh}},
\Ham_\M)$ and continue to denote $L_{\pi_{\mbox{\tiny nh}}} = graph
(\pi_{\mbox{\tiny nh}}^\sharp)$. The gauge transformation of
$\pi_{\mbox{\tiny nh}}$ associated to a 2-form $B$ on $\M$ gives
\begin{equation} \tau_B(L_{\pi_{\mbox{\tiny nh}}}) = \{(X, \alpha +
{\bf i}_X B ) \in T\M \oplus T^*\M \ : \ \pi_{\mbox{\tiny
nh}}^\sharp(\alpha) = X\}. \label{E:GaugedNH}
\end{equation}

First of all we are interested in knowing when the pair
$(-X_{\textup{\mbox{\tiny nh}}}, d\Ham_\M)$ is a section of $
\tau_B(L_{\pi_{\mbox{\tiny nh}}})$. On the other hand, we would also like to know whether the almost
Dirac structure $ \tau_B(L_{\pi_{\mbox{\tiny nh}}})$ corresponds to the
graph of a bivector field or not.

If the 2-form $B$ on $\M$ verifies   ${\bf
i}_{\mbox{$X_{\mbox{\tiny nh}}$}} \, B = 0$, then from equation
(\ref{E:GaugedNH}) we see that the pair $(-X_{\mbox{\tiny nh}},
d\Ham_\M)$ belongs to $\Gamma(\tau_B(L_{\pi_{\mbox{\tiny nh}}}))$.
Moreover, in view of  \eqref{Eq:GaugeSection}, the
gauge transformation of $\pi_{\mbox{\tiny nh}}$ by the 2-form $B$ has the form
\begin{equation} \tau_B(L_{\pi_{\mbox{\tiny nh}}}) = \{ (X,\alpha) \in
T\M \oplus T^*\M \ : \ X \in \C, \ {\bf i}_X (\Omega_\M - B)|_\C =
-\alpha |_\C \}. \label{E:GaugedOmega}
\end{equation} Thus the equations of motion
\eqref{E:Eqns_of_motion_Almost_Almost_Ham} are equivalently written
as
\begin{equation*} {\bf i}_{\mbox{$X_{\mbox{\tiny nh}}$}} \,(
\Omega_\C - B_\C) = (d\Ham_\M)_\C, \label{E:Omega0_modif}
\end{equation*}
where $B_\C$  is the point-wise restriction of $B$ to $\C$.

Therefore, as a particular case of Corollary \ref{C:NonDegSection},
we observe that if the section $
\Omega_\C - B_\C$ is non-degenerate then the gauge
transformation of $\pi_{\mbox{\tiny \textup{nh}}}$ associated to the
2-form $B$ is again a bivector field $\pi_{\mbox{\tiny
\textup{nh}}}^B$. It follows from \eqref{Eq:GaugedBivector} that
 the non-degeneracy of $ \Omega_\M - B$ on $\C$ is equivalent
to  the invertibility of the endomorphism $(\textup{Id} - B^\flat
\circ \pi^\sharp_{\mbox{\tiny nh}})$ on $T^*\M$ .
%
%

\bigskip

\begin{definition} \label{D:dynamicalGauge} Let $(P, \pi)$ be an almost Poisson
manifold with a distinguished Hamiltonian function $H \in
C^\infty(P)$. Given a 2-form $B$ on $P$, the gauge transformation of
$\pi$ associated to the 2-form $B$ is said to be a {\it dynamical gauge transformation} if
\begin{enumerate} \item[$(i)$] ${\bf i}_{X_H} B =0$, where $X_H$ is the (almost) Hamiltonian vector field
associated to $H$ and
\item[$(ii)$] $\tau_B (graph(\pi^\sharp))$ corresponds to
the graph of a new bivector $\pi^B$, i.e.,  the endomorphism
$(\textup{Id} - B^\flat\circ \pi^\sharp)$ on $T^*P$ is invertible.
\end{enumerate}
\end{definition}

Of course we are interested in dynamical gauge transformations of the nonholonomic bracket
where the distinguished Hamiltonian function is $\Ham_\M$ and the corresponding (almost)
Hamiltonian vector field is $\Xnh$.

Note that if $\pi$ is regular,  by equation
\eqref{Eq:GaugeSection}, the gauge transformation defined by $B$ is determined by
the restriction $B_F$ of $B$ to $F$ where $F:=\pi^\sharp(T^*P)$. Then condition ($i$) of the above Definition
is equivalent to ${\bf i}_{X_H} B_F =0$.
%
%

\begin{remark} \label{Rem:GeneralGauge} The definition
of an {\it affine almost Poisson bracket} for a nonholonomic system
made in \cite{Naranjo2008}  corresponds to a dynamical gauge
transformation of the nonholonomic bracket by a 2-form
$B=-\iota^*\Omega_0$ where $\Omega_0$ is a semi-basic form 2-form on
$T^*Q$. The proof is analogous to that of item ($ii$) of Proposition
\ref{L:nhbracket}. In this case, the hypothesis that $\Omega_0$ is
semi-basic implies that the condition ($ii$) of Definition
\ref{D:dynamicalGauge} is satisfied (see Proposition
\ref{C:NHGaugeSemibasic} below).
%
\end{remark}

\medskip

After this discussion, we observe that it is more appropriate to describe a
nonholonomic system by the triple $(\M, \mathfrak{F},
\Ham_\M)$ where $\mathfrak{F}$ is the family of  bivector fields that are
related to $\pi_{\mbox{\tiny nh}}$ through a dynamical gauge transformation.
  Notice that  $\C$ is the characteristic distribution of any bivector field in $\mathfrak{F}$.
 Thus, the (almost) Hamiltonian vector fields defined by the corresponding  brackets satisfy the nonholonomic constraints.
It follows that  the bivector fields in the family $\mathfrak{F}$ are almost
Poisson bivectors in a ``strong'' sense since the non-integrability of $\C$ prevents  them
from being twisted or conformally Poisson. Our interest in considering this big
family of brackets relies on reduction. In the presence of symmetries, the bracket that
Hamiltonizes the reduced equations may arise as the reduction of a member of $\mathfrak{F}$
that is not necessarily $\pi_{\mbox{\tiny nh}}$.

Since the distribution $\C$ is regular we observe
\begin{corollary} {\bf [of Theorem \ref{T:GaugedRelation}]}.
All bivectors with
characteristic distribution equal to $\C$ are gauge related (in particular, gauge related to the nonholonomic bivector $\pi_{\mbox{\tiny \textup{nh}}}$).
\end{corollary}

\medskip

We finish this section by discussing some cases for which the
second condition in Definition \ref{D:dynamicalGauge} is satisfied.
  Recall that $\M\subset T^*Q $ is a vector bundle over $Q$. We have

\begin{proposition}\label{C:NHGaugeSemibasic}
If $B$ is a semi-basic 2-form on $\M$, then the gauge transformation
of $\pi_{\mbox{\tiny \textup{nh}}}$ associated to $B$ corresponds again to a
bivector field.
\end{proposition}

\begin{proof}
The graph of $\pi_{\mbox{\tiny{nh}}}^\sharp$ is an almost Dirac
structure  corresponding to the pair $(\C, \Omega_\C)$ in the sense
of Proposition \ref{L:Dirac2Section} (see Proposition
\ref{L:nhbracket} ($i$)). Thus, in view of \eqref{Eq:GaugeSection},
the gauge transformation of $\pi_{\mbox{\tiny{nh}}}$ is the almost
Dirac structure corresponding to the pair  $(\C,
(\Omega_\M-B)|_\C)$. It is shown in \cite{Naranjo2008} that if $B$ a
is a semi-basic 2-form, then the point-wise restriction of $(\Omega_\M - B)$ to $\C$
 is non-degenerate. Thus, by Corollary
\ref{C:NonDegSection}, $\tau_B(L_{\pi_{\mbox{\tiny nh}}})$
corresponds to the graph of a bivector.
%
%
%
\end{proof}

In fact this Proposition is a special case of the following result:

\begin{proposition}\label{P:GaugeSemibasic}
Let $P \rightarrow Q$ be a vector bundle and $\pi$ a regular almost
Poisson bivector on $P$. If for all semi-basic 1-forms $\alpha$ on
$P$ the vector field $\pi^\sharp(\alpha)$ is  vertical,
then the gauge transformation of $\pi$ associated to a semi-basic
2-form $B$ corresponds again to a bivector.
\end{proposition}

\begin{proof}
Consider local bundle coordinates $({\bf q}, {\bf p})\in U\times V\subset \R^n\times \R^m$,
on  $P$ such that ${\bf q}$  are local
coordinates on the base manifold $Q$.  Since
 $\pi^\sharp(d{\bf q})$ is  vertical and $B$ is semi-basic we obtain
\begin{eqnarray*} B^\flat \circ \pi^\sharp (d{\bf q}) &=& 0   \\
B^\flat \circ \pi^\sharp (d{\bf p}) &=& b({\bf q}, {\bf p})\, d{\bf q}, \end{eqnarray*}
where $b({\bf q}, {\bf p})$ denotes the $m\times n$ matrix with entries
$b_{aj}({\bf q}, {\bf p})= \langle B^\flat \circ \pi^\sharp (dp_a),
\frac{\partial}{\partial q^j} \rangle$, for $a=1,...,m$, and
$j=1,...,n$.
 Thus, the matrix representation
of the endomorphism $(\textup{Id} - B^\flat\circ \pi^\sharp)$ on
$T^*P$ is $$\left(\begin{array}{cc} \textup{Id}_{n\times n} &
-b({\bf q}, {\bf p})^T
\\ 0 & \textup{Id}_{m\times m} \end{array} \right).$$
This matrix has full rank and hence $(\textup{Id} - B^\flat \circ
\pi^\sharp ) : T^*P \rightarrow T^*P$ is invertible.
\end{proof}

%

\begin{remark} \label{R:GeneralDynGauge} It is interesting for future work to  drop the condition ($ii$) in Definition
\ref{D:dynamicalGauge}
that  requires $\tau_B(graph(\pi^\sharp_{\mbox{\tiny
nh}}))$   to  define a bivector field. In this case, the family
$\mathfrak{F}$ consists of all the  almost Dirac
structures that are gauge related to $L_{\mbox{\tiny nh}}$ and that describe the nonholonomic
dynamics. In this sense, the Hamiltonization of the problem is achieved if the
reduction of a member of $\mathfrak{F}$ is  a Dirac structure (but not necessarily a
Poisson structure). This approach requires the consideration of a general reduction
scheme for almost  Dirac structures. However, we are unaware of
any examples of nonholonomic systems that justify the need of such a general framework.
 \end{remark}

\medskip

\subsection{Reduction by a group of symmetries}
\label{SS:Reduction-general}

We now add symmetries to the problem and perform the reduction. Our interest
from the point of view of Hamiltonization is to find a bivector field
in the family $\mathfrak{F}$ whose reduction is either Poisson or conformally Poisson
(see Section \ref{Ss:ConformalFactor}).

%

 Let $G$ be a Lie group acting
 freely and properly on $Q$. We say that that $G$ is a {\it symmetry
} of the nonholonomic system if the  lifted action to $TQ$ is free and
proper, and leaves the constraint distribution $\D
\subset TQ$ and the Lagrangian $\mathcal{L} :TQ \rightarrow \R$
invariant.

Denote by  $\Psi: G \times T^*Q \rightarrow T^*Q$ the cotangent lift of
the action to $T^*Q$. If $G$ is a symmetry for our nonholonomic
system, then $\Psi$ leaves both the constraint submanifold $\M$ and the
Hamiltonian $\Ham:T^*Q \rightarrow \R$ invariant. We continue to denote by $\Psi$
the restricted action to $\M$.

One can show that the tangent lift of  $\Psi$ to $T\M$,
preserves the distribution $\C$ and the section
$\Omega_\C$. As a consequence, if $G$ is a symmetry group for our
nonholonomic system, then the action $\Psi$
preserves the standard nonholonomic bracket $\{ \cdot,
\cdot\}_{\mbox{\tiny nh}}$. That is, for $f, g \in C^\infty(\M)$, we have \begin{equation*} \{f \circ \Psi , g \circ
\Psi \}_{\mbox{\tiny nh}} = \{f , g \}_{\mbox{\tiny nh}}\circ \Psi.
\label{E:NHBracketPreserved} \end{equation*} By freeness and properness of the action,
the reduced space $\RR := \M / G$
is a smooth manifold  and the orbit projection map $\rho: \M\to \RR$ is a surjective
submersion.  Notice
that $\RR$ inherits a vector bundle structure from $\M$ over the \emph{shape space} $Q/G$.
Moreover, $\RR$ is equipped with the reduced nonholonomic
bracket $\{ \cdot, \cdot\}_{\mbox{\tiny red}}$ that is characterized  by
\begin{equation} \{f , g \}_{\mbox{\tiny red}}\circ \rho(m) := \{f \circ
\rho , g \circ \rho\}_{\mbox{\tiny nh}} (m) \qquad \mbox{for } m \in \M
\mbox{ and } f, g \in C^\infty(\RR). \label{E:NHBracketReduction}
\end{equation}
The corresponding bivector field will be denoted $\pi_{\mbox{\tiny red}}$.
The reduced nonholonomic bracket describes the reduced dynamics in the sense that the
nonholonomic vector field $X_{\mbox{\tiny nh}}$ is $\rho$-related to
the (almost) Hamiltonian vector field $X_{\Ham_{\RR}} = \{ \cdot , \Ham_\RR\}_{\mbox{\tiny red}}$ associated to
the reduced Hamiltonian $\Ham_\RR$ defined by the condition $\Ham_\M
= \Ham_\RR \circ \rho$.

%

\begin{remark} The reduction of nonholonomic systems performed by
Bates and Sniatycki in \cite{BS93} shows that it is possible to
define a 2-form $\omega_{\mbox{\tiny red}}$ on $\RR$ which is
non-degenerate along a distribution $\bar \C \subset T\RR$. The
definition of $\bar \C$ is given by $\bar \C:= T\rho (U)$ where $U
:= \C \cap (\C \cap V)^{\Omega_\M}$ and $V$ is the distribution on
$\M$ tangent to the orbits of $G$. In fact, the pair $(\bar \C,
\omega_{\mbox{\tiny red}})$ is just the pair associated to the
almost Dirac structure given by the bivector field $\pi_{\mbox{\tiny
red}}$ in the sense of Corollary \ref{C:NonDegSection}.
\end{remark}

The analysis of the  reduction of the bivector fields $\pi_{\mbox{\tiny nh}}^B$ in the family $\frak{F}$ that
are  related to $\pi_{\mbox{\tiny nh}}$ by a dynamical gauge $B$ follows from:

\begin{proposition} \label{P:GaugedPreserved}
Let $(P,\pi)$ be an almost Poisson manifold and $B$ a 2-form on $P$
such that the endomorphism $(\textup{Id} - B^\flat \circ \pi^\sharp)
: T^*P \rightarrow T^*P$ is invertible. If the Lie group $G$, acting
freely and properly on $P$, preserves the almost Poisson structure
$\pi$ and leaves $B$ invariant, then $G$ preserves the bivector
field $\pi^B$ obtained by the gauge transformation of $\pi$
associated to $B$.
\end{proposition}

\begin{proof} In view of equation (\ref{Eq:GaugedBivector}), we see
that $(\pi^B)^\sharp$ is a composition of invariant maps and we
conclude that $\pi^B$ is invariant as well.
\end{proof}

As a direct consequence of this Proposition we have:
\begin{proposition}
\label{P:reduction-gauge}
If  $G$ is a symmetry group of the nonholonomic system $(\M,
\pi_{\mbox{\tiny \emph{nh}}}, \Ham_\M)$ and  $B$ is a $G$-invariant dynamical
gauge on $\M$, then  $\{ \cdot, \cdot\}_{\mbox{\tiny \emph{nh}}}^B$ is $G$ invariant.
In particular there is a reduced bivector field $\pi_{\mbox{\tiny \emph{red}}^B}$ on the reduced space
$\RR$ that determines a well defined bracket $\{ \cdot , \cdot \}_{\mbox{\tiny \emph{red}}^B}$ on $\RR$
 satisfying
$$
\{f, g\}_{\mbox{\tiny ${\textup{red}}$}^B} \circ \rho(m) = \{f \circ \rho, g \circ \rho \}^B_{\mbox{\tiny \emph{nh}}} (m),
\qquad \mbox{for } f, g \in C^\infty(\RR).
$$
Moreover, the reduced bracket $\{ \cdot , \cdot \}_{\mbox{\tiny \emph{red}}^B}$ also describes
the reduced dynamics in the sense that $X_{\Ham_{\RR}} = \{ \cdot , \Ham_\RR\}_{\mbox{\tiny \emph{red}}^B}$.
\end{proposition}

There is a good reason why we did not denote the reduced bivector
field $\pi_{\mbox{\tiny {red}}^B}$ by $\pi_{\mbox{\tiny {red}}}^B$ and that is
that in general the reduced bivector fields $\pi_{\mbox{\tiny {red}}}$ and $\pi_{\mbox{\tiny {red}}^B}$
need \emph{not} be gauged related.  We shall see this explicitly
with the analysis of our mechanical examples (Remark \ref{R:reduced_not_gauged}).

\medskip

\begin{remark} Since the almost Dirac structure $L_{\mbox{\tiny \textup{nh}}}$ given in (\ref{E:DiracNH}) is
the graph of a bivector (see Proposition \ref{L:nhbracket} ($i$)),
then the reduction of $L_{\mbox{\tiny \textup{nh}}}$ as an almost Dirac structure is simply the
reduction of the bivector $\pi_{\mbox{\tiny nh}}$ in the classical
way. The same observation is valid for the reduction of the almost
Dirac structure $\tau_B(L_{\pi_{\mbox{\tiny nh}}})$ defined in
(\ref{E:GaugedOmega}) since $\tau_B(L_{\pi_{\mbox{\tiny nh}}})$ is
also the graph of a bivector ($(\Omega_\M -B)$ is non-degenerate on
$\C$).
\end{remark}

\subsubsection*{The $G$-Chaplygin case.}
If the reduced bivector fields  $\pi_{\mbox{\tiny {red}}}$ and $\pi_{\mbox{\tiny {red}}^B}$  happen to be everywhere
non-degenerate, then they are  gauge-related.
This is the scenario that one finds after reduction of external symmetries of $G$-Chaplygin systems
(see \cite{EhlersKoiller}). These systems are characterized by the property that the tangent space to the
orbits of the symmetry group exactly complements the constraint
distribution on the tangent space $TQ$ of the configuration manifold
$Q$.

In this case there exist unique non-degenerate 2-forms
$\Omega_{\mbox{\tiny {red}}}$  and $\Omega_{\mbox{\tiny {red}}^B}$
on $\RR$ satisfying
\begin{equation*}
\pi_{\mbox{\tiny {red}}}^\sharp \circ \Omega_{\mbox{\tiny
{red}}}^\flat =\textup{Id}, \qquad \pi_{\mbox{\tiny {red}}^B}^\sharp
\circ
\Omega_{\mbox{\tiny {red}}^B}^\flat =\textup{Id}, 
\end{equation*}
 where as usual, $\Omega_{\mbox{\tiny
{red}}}^\flat(X) = - {\bf i}_{X}\Omega_{\mbox{\tiny {red}}}$ and
$\Omega_{\mbox{\tiny {red}}^B}^\flat(X) = -{\bf
i}_{X}\Omega_{\mbox{\tiny {red}}^B}$ for all $X\in
\mathfrak{X}(\RR)$.
In particular, the reduced equations can be
written as:
 \begin{equation*}
{\bf i}_{X_{\Ham_{\RR}}}\Omega_{\mbox{\tiny {red}}}={\bf i}_{X_{\Ham_{\RR}}}\Omega_{\mbox{\tiny {red}}^B}=d\Ham_\RR,
\end{equation*}
which gives  different almost symplectic formulations of the reduced
equations (compare with the results in \cite{Hoch}). The bivector
fields $\pi_{\mbox{\tiny {red}}}$ and $\pi_{\mbox{\tiny {red}}^B}$
are gauge related via the 2-form $B_{\mbox{\tiny {red}}}:=
\Omega_{\mbox{\tiny {red}}}-\Omega_{\mbox{\tiny {red}}^B}$ on $\RR$.
Moreover, the 2-forms $\Omega_{\mbox{\tiny {red}}}$ and
$\Omega_{\mbox{\tiny {red}}^B}$ satisfy
\begin{equation*}
(\rho^* \Omega_{\mbox{\tiny red}} )|_\C = \Omega_\C \qquad
\mbox{and} \qquad (\rho^* \Omega_{{\mbox{\tiny red}}^B}) |_\C =
(\Omega_\M-B)_\C,
\end{equation*}
and thus $B_{\mbox{\tiny red}}$ verifies that $(\rho^*B_{\mbox{\tiny
red}}) |_\C = B_\C.$ In other words, we have the following
commutative diagram $$\xymatrix{
L_{\pi_{\mbox{\tiny nh}}}\ar[d]^{\rho}\ar[rr]^{\tau_B}&& L_{\pi_{\mbox{\tiny nh}}^B}\ar[d]^{\rho}\\
L_{\pi_{\mbox{\tiny red}}} \ar[rr]^{\tau_{B_{\mbox{\tiny red}}}} &&
L_{\pi_{{\mbox{\tiny red}}^B} }.}$$ We stress that one does not have
such a commutative diagram in more general situations, see Remark
\ref{R:reduced_not_gauged}.

\subsection{Hamiltonization} \label{Ss:ConformalFactor}

 We continue with the notation of the previous sections
and suppose that $G$ is a symmetry group for our nonholonomic system. The solutions of the reduced equations
on the reduced space $\RR=\M/G$ are the integral curves of the reduced vector field $X_{\Ham_{\RR}}$
and preserve the reduced Hamiltonian $\Ham_\RR$.

 The issue
of Hamiltonization in our context concerns answering the question of
whether the reduced vector field $X_{\Ham_{\RR}}$ on $\RR$ is
Hamiltonian. Our candidates for the Hamiltonian structure come from the reduction
of the (invariant) bivector fields that belong to  the family $\frak{F}$.

It  turns out that the above condition is too restrictive.
We relax it by asking that the vector field $X_{\Ham_{\RR}}$ can be
rescaled by a basic\footnote{We use the term basic with respect to the fibered structure of $\RR$ inherited from
$\M$. That is  $\varphi=\tilde \varphi \circ \tau$ where $\tau :\RR\rightarrow Q/G$ is the bundle projection and
$\tilde \varphi:Q/G \rightarrow \R$.} positive function $\varphi :\RR\rightarrow \R$ in such
a way that the resulting vector field $\varphi \, X_{\Ham_\RR}$ is
Hamiltonian.

In view of Proposition \ref{P:reduction-gauge}, for any
$G$-invariant bivector field $\pi_{\mbox{\tiny nh}}^B$ belonging to $\frak{F}$,
the rescaled vector field $\varphi\,
X_{\Ham_{\RR}}$ satisfies
\begin{equation}
\label{E:Hamiltonization} (\varphi \pi_{\mbox{\tiny {red}}^B} )^\sharp (d \Ham_\RR) = - \varphi X_{\Ham_{\RR}}.
\end{equation}
Hence, we are interested in finding a bivector field $\pi_{\mbox{\tiny {red}}^B}$ satisfying
\begin{equation}
\label{E:ConfPoisson}
[\varphi \pi_{\mbox{\tiny {red}}^B}, \varphi \pi_{\mbox{\tiny {red}}^B} ] =0,
\end{equation}
that is, we want to  find $\pi_{\mbox{\tiny {red}}^B}$  conformally Poisson (see Definition \ref{D:Basics}).

\begin{definition}
\label{D:Hamiltonization}
If there exist an invariant bivector field  $\pi_{\mbox{\tiny nh}}^B$ belonging to $\frak{F}$ and a strictly positive,
basic function $\varphi:\RR\to \R$ such that \eqref{E:Hamiltonization} and \eqref{E:ConfPoisson} hold,
we say that the nonholonomic system is \emph{Hamiltonizable}. Moreover, we say that the reduced equations
are Hamiltonian in the new time $\tau$ defined by $d\tau =\frac{1}{\varphi}dt$ (see discussion below).
\end{definition}

In Section \ref{S:Hamiltonization} we will extend the discussion of  Section \ref{S:Examples} and
show that all generalized rolling systems are Hamiltonizable. The table \eqref{E:Table-brackets}
in Section \ref{S:Hamiltonization}
shows how different scenarios of the Hamiltonization scheme described above are realized according
to the rank of the matrix $A$.

\subsubsection*{Time reparametrizations}

It is common in the literature to interpret the
rescaling of the vector field $X_{\Ham_{\RR}}$  by the basic positive function
$\varphi$ as a \emph{nonlinear time reparametrization}. One
argues as follows, let  $c(t)\in \RR$ be a flow line of $X_{\Ham_{\RR}}$
(i.e. $\frac{d c}{dt}(t)= X_{\Ham_{\RR}} (c(t))$).  Introduce the new time
$\tau$ by integrating the relation
\begin{equation*}
d\tau =\frac{1}{\varphi( c (t))} \, dt.
\end{equation*}
Since $\varphi >0$, the correspondence between $t$ and $\tau$ is one-to-one
and one can express $t$ as a function of $\tau$. The curve  $\tilde
c (\tau):= c(t(\tau))$ is checked to be a flow line of $\varphi X_{\Ham_{\RR}}$ (i.e.
$\frac{d \tilde c}{d\tau }(\tau)= \varphi(\tilde c (\tau)) X_{\Ham_{\RR}} (\tilde
c(\tau))$). This interpretation of the rescaling is quite subtle.
The definition of $\tau$ depends on the particular flow line $c(t)$,
so different initial conditions induce  different
reparametrizations. It is therefore not possible to interpret the
time rescaling as a ``global" operation. This contrasts with the
natural procedure of multiplying the vector field $X_{\Ham_{\RR}}$ by the
positive function $\varphi$.

\begin{remark}
One might wonder why we only care about basic and not arbitrary functions $\varphi:\RR\to \R^+$.
A detailed answer to this question would involve a careful study of the structure of the equations of motion that
would gear us away from the main subject of this paper. We refer the reader to \cite{FedorovJovan} where one can find a  very
good discussion on the Hamiltonization of $G$-Chaplygin systems. We simply mention that
physically, the fact that $\varphi$ is basic  means that it is independent of the momentum variables and
only depends on the (reduced) configuration variables. Thus, the time reparametrization changes the speed at
which the trajectories are traversed depending on the position of the system but independently of the velocity itself.
It is shown in   \cite{FedorovJovan}  how in order to obtain Darboux coordinates for the reparametrized system, one can
keep the same (reduced) configuration variables but should rescale the   momenta by
$\frac{1}{\varphi}$. Since $\varphi$ is basic, the rescaled momenta continue to depend linearly on the velocities.
%
\end{remark}


\subsubsection*{Measure preservation}

Hamiltonization is strongly related to the existence of invariant
measures. Suppose for simplicity that  we are dealing with a $G$-Chaplygin system. As mentioned before,
in this case the bivector field $\pi_{\mbox{\tiny {red}}^B}$ is everywhere non-degenerate
and the reduced equations can be written as:
\begin{equation*}
{\bf i}_{X_{\Ham_\RR}}\Omega_{\mbox{\tiny {red}}^B} =d\Ham_\RR,
\end{equation*}
where $\Omega_{\mbox{\tiny {red}}^B}$ is the non-degenerate 2-form on $\RR$  induced by $\pi_{\mbox{\tiny {red}}^B}$.
It follows that the scaled vector field $\varphi X_{\Ham_\RR}$ satisfies
\begin{equation*}
{\bf i}_{\, \varphi X_{\Ham_\RR}}\left ( \mbox{\small $\frac{1}{\varphi}$}\Omega_{\mbox{\tiny {red}}^B}\right ) =d\Ham_\RR.
\end{equation*}
Hence, Hamiltonization in this setting amounts to finding a positive function $\varphi$ such that the
2-form $\mbox{\small $\frac{1}{\varphi}$}\Omega_{\mbox{\tiny {red}}^B}$ is closed (which under
our hypothesis,  is of course
equivalent to \eqref{E:ConfPoisson}). Suppose for a moment that this is the case so
$(\RR, \mbox{\small $\frac{1}{\varphi}$}\Omega_{\mbox{\tiny {red}}^B})$ is a symplectic manifold.
It follows from Liouville's theorem that the vector field $\varphi X_{\Ham_\RR}$ preserves the
symplectic volume
 $  \left ( \mbox{\small $\frac{1}{\varphi}$}\Omega_{\mbox{\tiny {red}}^B}  \right )^m$ where $m=\frac{1}{2}\dim \RR$.
 Therefore, the volume form $(\mbox{\small $\frac{1}{\varphi}$})^{m-1} (\Omega_{\mbox{\tiny {red}}^B})^m $
 is preserved by the vector field $X_{\Ham_\RR}$.


The above argument shows that a Hamiltonizable $G$-Chapligyn system
possesses an invariant measure. One might wonder if the reciprocal
statement is true, namely, if any $G$-Chaplygin system with an invariant
measure is Hamiltonizable. The celebrated Chaplygin's reducing
multiplier theorem \cite{Chapligyn_reducing_multiplier} demonstrates
that the answer is positive if $m=2$. For $m>2$, a
characterization of the systems for which this is true  is an open
problem. Interesting examples where this holds for arbitrary values of $m$ have been
found by Fedorov and Jovanovic in the study of the multidimensional Veselova problem \cite{FedorovJovan}.
See also the discussion in \cite{EhlersKoiller} where a candidate for the conformal factor
$\mbox{\small $\frac{1}{\varphi}$}$ is given under the hypothesis that there exists a preserved measure,
and \cite{Fernandez} where a set of coupled first order partial differential equations
for the multiplier $\mbox{\small $\frac{1}{\varphi}$}$ are given.

In the case where the  nonholonomic system is Hamiltonizable but the corresponding
bivector field $\pi_{\mbox{\tiny {red}}^B}$ is degenerate at some points in $\RR$, one can repeat the above argument
to conclude that the reduced system preserves a measure on every leaf of the symplectic foliation of $\RR$
corresponding to the Poisson bivector field $\varphi \pi_{\mbox{\tiny {red}}^B}$.
However, this does \emph{not} imply the existence of
a smooth invariant measure on $\RR$ (and this is the motivation for Fernandez, Mestdag, and Bloch
\cite{Fernandez} to talk about \emph{Poissonization}). An example of this situation is given by the reduction
of the Chaplygin sleigh (see the discussion in \cite{Naranjo2007} ) that exhibits asymptotic dynamics
that contravene the existence of a global invariant measure. The problem of the existence
of a global invariant measure in this case is most naturally attacked by considering
the \emph{modular class} of the Poisson manifold $(\RR, \varphi \pi_{\mbox{\tiny {red}}^B})$, see \cite{Weinstein}.


\section{Back to the Examples: Hamiltonization and Integrability} \label{S:Hamiltonization}

In this section, we analyze the generalized rolling systems presented in Section
\ref{S:Examples} using  the geometric framework that was developed  in the previous sections.
In subsection \ref{SS:Geom_brackets} we provide the geometric interpretation for the
brackets $\{\cdot , \cdot \}_{\mbox{\tiny Rank{\it j}}}$ and $\{\cdot , \cdot \}_{\mbox{\tiny Rank{\it j}}}'$ presented
in Section \ref{S:Examples}.
In  \ref{SS:Hamiltonization-Integrability} we consider the Hamiltonization and integrability
of generalized rolling systems in detail and finally, in \ref{Ss:TwistedMechanics} we explicitly
show that the brackets $\{\cdot , \cdot \}_{\mbox{\tiny Rank1}}$ and $\{\cdot , \cdot \}_{\mbox{\tiny Rank{2}}}'$
are twisted Poisson. To our knowledge, this is the first time that the appearance of such structures is made
explicit in the field of nonholonomic mechanics.

\subsection{The geometry of the rigid bodies with generalized rolling constraints}
\label{SS:Geom_brackets}

We begin by computing the  nonholonomic bracket  for the motion of a
rigid body subject to generalized rolling constraints as introduced
 in Section \ref{S:Examples}.

\subsubsection*{Nonholonomic bracket via the non-degenerate 2-section}

Consider again, as in Section \ref{S:Examples}, the motion of a rigid
body in space subject to a generalized rolling constraint as in (\ref{E:general_constraint}).
That is, the constraint relates the linear and the angular
velocities of the body ${\bf x}=r A \vecom=rAg\vecOm,$ where the matrix $A$
satisfies any of the conditions of Definition \ref{condA} and $\vecom, \; \vecOm$ is the angular velocity
written in space and body coordinates, respectively.

Recall that the configuration space for the system is $Q=\operatorname{SO}(3)\times \R^3$. Denote by
$\vecL$ (respectively, $\vecR$) the left (respectively, right) Maurer-Cartan form on $\operatorname{SO}(3)$.
Upon the identification of the Lie algebra $\so(3)$ with $\R^3$ by the hat map \eqref{E:hat-map} we think
of $\vecL$ and $\vecR$ as $\R^3$ valued 1-forms on $\operatorname{SO}(3)$. For a tangent vector
 $v_g\in T_g\operatorname{SO}(3)$ we have
 \begin{equation*}
\vecom=\vecR(g)(v_g), \qquad \vecOm=\vecL(g)(v_g),
\end{equation*}
where $\vecom$ (respectively, $\vecOm$) denotes the angular velocity vector written
in space (respectively, body) coordinates as discussed in Section \ref{S:Examples}.

The Maurer-Cartan forms $\vecL$ and $\vecR$ are related by $\vecL(g)=g^{-1}\vecR(g)$ and satisfy the
well-known Maurer-Cartan equations
\begin{equation*}
d\vecR=[\vecR,\vecR], \qquad d\vecL=-[\vecL,\vecL],
\end{equation*}
where $[\cdot, \cdot]$ is the commutator in the Lie algebra.
For the rest of the section we will use three dimensional vector algebra notation in our calculations with differential forms
and vector fields. In our convention, the
scalar product of differential forms should always be interpreted as a wedge product (and is thus anti-commutative!). The Maurer-Cartan equations
take the form
\begin{equation}
\label{E:Maurer-Cartan-Cross-Product}
d\vecR=\frac{1}{2} \vecR \times \vecR, \qquad d\vecL=-\frac{1}{2} \vecL \times \vecL,
\end{equation}
where ``$\times$" denotes the standard vector product in $\R^3$.

The constraint distribution $\D$, defined by the generalized rolling constraints,
can be expressed in the terminology of subsection \ref{S:Nonholonomic_Systems}
as the annihilator of the $\R^3$-valued 1-form $\vecep$ on $Q$ given by
\begin{equation*}
\vecep=d{\bf x}-rA\vecR=d{\bf x}-rAg\vecL.
\end{equation*}

We consider the (global) moving co-frame $\{\vecL, d{\bf x}  \}$ for $T^*Q$ that defines fiber coordinates $({\bf M}, {\bf p})$ in
the following sense. A co-vector $\alpha_q\in T^*_qQ$ is
written uniquely as $\alpha_q = {\bf M}\cdot \vecL + {\bf p}\cdot d{\bf x}$, for a certain $({\bf M}, {\bf p})\in \R^3\times \R^3$. The Legendre
transform $\L:TQ\to T^*Q$ associated to the kinetic energy Lagrangian \eqref{E:lag-rigidbody} is defined by the rule:
\begin{equation*}
{\bf M}=\I \vecOm, \qquad {\bf p}=m\dot {\bf x}.
\end{equation*}
Physically, ${\bf p}$ is the linear momentum of the body while ${\bf M}$ is the angular momentum of the body
about the center of mass written in body coordinates.

In order to deal with the constraints, it is more convenient to work with the global moving co-frame
 $\{\vecL, \vecep   \}$ for $T^*Q$.  We denote by $({\bf K},{\bf u})$ the fiber coordinates defined by this
 co-frame. Putting $ {\bf K}\cdot \vecL +{\bf u}\cdot \vecep = {\bf M}\cdot \vecL +{\bf p}\cdot d{\bf x}$ implies
 \begin{equation*}
{\bf K}={\bf M}+rg^TA^T{\bf p}, \qquad  {\bf u}={\bf p}.
\end{equation*}
Along the constraint submanifold $\M=\L(\D)$ we have ${\bf p}=mrAg\vecOm$ so
\begin{equation*}
{\bf K}=\I\vecOm +mr^2g^TA^TAg\vecOm,
\end{equation*}
which is the expression for the kinetic momentum obtained in \eqref{E:Kinetic_Momentum}. Notice that $\M$ is a vector
bundle over $Q$ and that
 ${\bf K}$ is a natural coordinate for the fibers of $\M$. In what follows we will use
 the components of $g, {\bf x}$, and  ${\bf K}$, as redundant coordinates on $\M$.

Denote by ${\bf X}^{\mbox{\tiny L}}=(X_1^L,X_2^L,X_3^L)$ the moving frame of $\operatorname{SO}(3)$
that is dual to $\vecL=(\lambda_1,\lambda_2,\lambda_3)$. The components of ${\bf X}^{\mbox{\tiny L}}$ are the
left invariant vector fields on $\operatorname{SO}(3)$ obtained by left extension of the canonical
basis of $\R^3$. Along the points of the constraint subbundle $\M$, the
 non-integrable distribution $\C$ defined  in
(\ref{E:C}) is given by
\begin{equation}
\label{E:basis_of_C}
\C=\mbox{span}\left \{ \, {\bf X}^{\mbox{\tiny L}} + rg^TA^T\frac{\partial}{\partial {\bf x}} \, , \, \frac{\partial}{\partial {\bf K}} \, \right \}.
\end{equation}

The canonical 2-form $\Omega_Q$ on $T^*Q$ is given by
\begin{equation*}
\begin{split}
\Omega_Q&=-d({\bf M}\cdot \vecL+ {\bf p}\cdot d{\bf x}) =-d( {\bf K}\cdot \vecL +{\bf p}\cdot \vecep) \\
&= \vecL\cdot d{\bf K} -{\bf K}\cdot d\vecL  -{\bf p}\cdot d\vecep +  \vecep \cdot d{\bf p},
\end{split}
\end{equation*}
where ``$\cdot$" denotes the usual scalar product in $\R^3$.

To compute $d\vecep$ we use the identity $dg=g\hat{\vecL}$ where $\; \hat{} \;$ denotes the hat map \eqref{E:hat-map}.
Using the Maurer-Cartan equations \eqref{E:Maurer-Cartan-Cross-Product} we get
\begin{equation*}
\begin{split}
d\vecep &=-rA(dg)\vecL -rAg\, d\vecL = -rAg (\vecL \times \vecL) -rAg\, d\vecL = rAg \, d\vecL.
\end{split}
\end{equation*}
Therefore,
\begin{equation*}
\begin{split}
\Omega_Q &= \vecL\cdot d{\bf K} -{\bf K}\cdot d\vecL  -{\bf p}\cdot (rAg \, d\vecL)  +  \vecep \cdot d{\bf p} \\
&= \vecL\cdot d{\bf K} -( {\bf K}+ rg^TA^T{\bf p}) \cdot d\vecL  +  \vecep \cdot d{\bf p}.
\end{split}
\end{equation*}

Let $\iota:\M \hookrightarrow T^*Q$ denote the inclusion. Since ${\bf p}=mrAg\vecOm$ along $\M$, we have
\begin{equation*}
\begin{split}
\Omega_\M:=\iota^*(\Omega_Q) &= \vecL\cdot d{\bf K} -( {\bf K}+
mr^2g^TA^TAg\vecOm) \cdot d\vecL  +  \iota^*(\vecep \cdot d{\bf p}).
\end{split}
\end{equation*}
Since $\vecep$ vanishes along the non-integrable distribution $\C$,
we get the following expression for the restriction $\Omega_\C$ of
$\iota^*(\Omega_Q)$ to $\C$:
 \begin{equation*}
\begin{split}
\Omega_\C &= \vecL\cdot d{\bf K} -( {\bf K}+ mr^2g^TA^TAg\vecOm) \cdot d\vecL .
\end{split}
\end{equation*}
With this expression for $\Omega_\C$ we are ready to show:

\begin{proposition}
\label{E:nonho-bracket-geometric}
The nonholonomic bracket $\{\cdot , \cdot \}_{\mbox{\emph{\tiny nh}}}$ on $\M$ for the generalized rolling system
is given in the redundant coordinates $\{ g_{ij},x_k, K_l \}$, $\, i,j,k,l=1,2,3$,  for $\M$ by
\begin{equation*}
\begin{split}
 \{x_i , K_l \}_{\mbox{\emph{\tiny nh}}}=r(Ag)_{il}, \qquad \{g_{ij} , K_l \}_{\mbox{\emph{\tiny nh}}} =-\varepsilon_{jl}^k\, g_{ik},
\qquad \{K_i , K_j \}_{\mbox{\emph{\tiny nh}}} =-\varepsilon_{ij}^l({\bf K}+mr^2(g^TA^TAg\vecOm))_l \, ,
\end{split}
\end{equation*}
with all other combinations equal to zero.
In the above formulas the Einstein convention of sum over repeated indices holds and $ \varepsilon_{ij}^k$ denotes
the alternating tensor, that equals $0$ if two indices are equal, it equals $1$ if $(i,j,k)$ is a cyclic permutation of $(1,2,3)$, and
it is equal to $-1$ otherwise. The entries of $(g,{\bf x})\in Q$ are denoted by $(g_{ij}, x_k)$ and those of ${\bf K}$ by $K_l$.
\end{proposition}

\begin{proof}
We will rely on the identity \eqref{E:DefBracket} that characterizes the associated bivector field $\pi_{\mbox{\tiny nh}}$ in
terms of  $\Omega_\C$. Contraction of $\Omega_\C$ by the elements in the basis of $\C$ given in \eqref{E:basis_of_C}
gives
\begin{equation*}
{\bf i}_{({\bf X}^{\mbox{\tiny L}} + rg^TA^T\frac{\partial}{\partial {\bf x}})}\Omega_\C =d{\bf K}-({\bf K}+ mr^2g^TA^TAg\vecOm)\times \vecL \, ,
\qquad {\bf i}_{\frac{\partial}{\partial {\bf K}}}\Omega_\C=-\vecL,
\end{equation*}
where we have again made use of the Maurer-Cartan equations \eqref{E:Maurer-Cartan-Cross-Product}. It follows that
\begin{equation*}
{\bf i}_{({\bf X}^{\mbox{\tiny L}} + rg^TA^T\frac{\partial}{\partial {\bf x}} -({\bf K}+ mr^2g^TA^TAg\vecOm)\times  \frac{\partial}{\partial {\bf K}})}\Omega_\C =d{\bf K}, \qquad {\bf i}_{-rAg \frac{\partial}{\partial {\bf K}}}\Omega_\C=rAg\vecL=d{\bf x}|_\C.
\end{equation*}
Therefore, according to \eqref{E:DefBracket} we get
\begin{equation}
\label{E:piflat_nonho_bracket}
\begin{split}
\pi_{\mbox{\tiny nh}}^\sharp (d{\bf K})&=-{\bf X}^{\mbox{\tiny L}} - rg^TA^T\frac{\partial}{\partial {\bf x}} +({\bf K}+ mr^2g^TA^TAg\vecOm)\times
 \frac{\partial}{\partial {\bf K}}, \\  &\pi_{\mbox{\tiny nh}}^\sharp (\vecL)= \frac{\partial}{\partial {\bf K}}, \qquad
 \pi_{\mbox{\tiny nh}}^\sharp (d{\bf x}) = rAg\frac{\partial}{\partial {\bf K}}\, .
\end{split}
\end{equation}
In addition, for any canonical vector ${\bf e}_i\in \R^3$ we have
$d(g^{-1}{\bf e}_i)=(g^{-1}{\bf e}_i)\times \vecL$, so
\begin{equation*}
 \pi_{\mbox{\tiny nh}}^\sharp (d(g^{-1}{\bf e}_i))=(g^{-1}{\bf e}_i)\times \frac{\partial}{\partial {\bf K}}.
\end{equation*}
The proof follows using the above formulas and  recalling that for
any $f,g \in C^\infty(\M)$, we have $\{f,g \}_{\mbox{\tiny nh}}=- df(\pi^\sharp_{\mbox{\tiny nh}}(dg) )$.

\end{proof}

Finally, we state without a formal proof that the nonholonomic vector field $\Xnh$ on $\M$ is given by
\begin{equation}
\label{E:Xnh}
\Xnh=\vecOm \cdot {\bf X}^{\mbox{\tiny L}} +r Ag\vecOm \cdot \frac{\partial}{\partial {\bf x}} + ({\bf K}\times \vecOm)\cdot \frac{\partial}{\partial {\bf K}}.
\end{equation}
The above expression can be shown by  taking into account the equations \eqref{E:Motion3}, the constraint
\eqref{E:general_constraint}, and the definition of $\vecOm$, or, alternatively, by computing the the
almost Hamiltonian vector field $X_{\Ham_\M}=-\pi_{\mbox{\tiny nh}}^\sharp (d\Ham_\M)$
corresponding to the Hamiltonian $\Ham_\M$ that coincides with expression
 \eqref{E:Ham}. The latter approach requires one to write $\vecOm$ in terms of ${\bf K}$ and $g$ as
was done in Section \ref{S:Examples} for the different values of the rank of $A$.

\subsubsection*{The gauge transformation of the nonholonomic bracket}

We will now construct a gauge transformation of the nonholonomic
bracket in the sense of subsection \ref{SS:GaugeTransf}. We are
interested in describing the same dynamics so we look for a 2-form
$B$ that defines a dynamical gauge transformation as introduced in
Definition \ref{D:dynamicalGauge}. In our case, the distinguished
Hamiltonian is $\Ham_\M$ that has $\Xnh$ as its associated almost
Hamiltonian vector field.

Following \cite{Naranjo2008}, we consider the bi-invariant
volume form $\nu$ on $\SO (3)$ oriented and scaled such that $\nu(X^L_1, X^L_2, X^L_3)=1$.
We
consider the natural extension of $\nu$ as a 3-form on $Q=\SO(3) \times \R^3$. Denote by
$\bar \nu \in \Omega^3(T^*Q)$ the 3-form given by $\bar \nu =
\upsilon^*\nu$ where $\upsilon: T^*Q \rightarrow Q$ is the canonical projection. We can write
$\bar \nu =\frac{1}{6} \vecL \cdot (\vecL\times \vecL)$.

Let $B$ be the 2-form on $\M$ given by
$$
B= mr^2 ({\bf i}_{\Xnh} \iota^* \bar \nu),
$$
where, as before,  $\iota: \M \hookrightarrow T^*Q$
is the inclusion. Note that $B$ is a semi-basic 2-form on $\M$ that vanishes
upon contraction with the nonholonomic (almost) Hamiltonian vector field $X_{\mbox{\tiny nh}}$.
Therefore, by Proposition
\ref{C:NHGaugeSemibasic}, we can perform a dynamical gauge
transformation of the nonholonomic bivector field $\pi_{\mbox{\tiny nh}}$
by the 2-form $B$ to obtain another  bivector field $\pi_{\mbox{\tiny nh}}^B$ that also describes  the dynamics
of our problem.

Using the Maurer-Cartan equations \eqref{E:Maurer-Cartan-Cross-Product}  and the expression \eqref{E:Xnh} for
$\Xnh$, we obtain
\begin{equation}
\label{Ex:2formB}
B= -mr^2 \vecOm\cdot d\vecL.
\end{equation}
To compute the bivector field $\pi_{\mbox{\tiny nh}}^B$ associated to the gauge transformation we
use equation \eqref{E:condition-gauge}. For an arbitrary one-form $\alpha$ on $\M$ we have
\begin{equation*}
(\pi_{\mbox{\tiny nh}}^B)^\sharp \left ( \alpha + {\bf i}_{\pi_{\mbox{\tiny nh}}^\sharp(\alpha)} B \right ) =
\pi_{\mbox{\tiny nh}}^\sharp(\alpha).
\end{equation*}
Setting $\alpha$ equal to $\vecL$ and $d{\bf x}$ and using \eqref{E:piflat_nonho_bracket} and \eqref{Ex:2formB}
we obtain
\begin{equation*}
(\pi_{\mbox{\tiny nh}}^B)^\sharp(\vecL)=\frac{\partial}{\partial {\bf K}}, \qquad
(\pi_{\mbox{\tiny nh}}^B)^\sharp(d{\bf x})=rAg\frac{\partial}{\partial {\bf K}}.
\end{equation*}
Similarly, putting $\alpha=d{\bf K}$ and noticing that
\begin{equation*}
 {\bf i}_{\pi_{\mbox{\tiny nh}}^\sharp(d{\bf K})} B=  -  {\bf i}_{ {\bf X}^{\mbox{\tiny L}} } B = mr^2 \vecOm \times \vecL,
\end{equation*}
we deduce
\begin{equation*}
(\pi_{\mbox{\tiny nh}}^B)^\sharp(d{\bf K})= -{\bf X}^{\mbox{\tiny L}} - rg^TA^T\frac{\partial}{\partial {\bf x}}
+({\bf K}+ mr^2(g^TA^TAg-E)\vecOm)\times  \frac{\partial}{\partial {\bf K}},
\end{equation*}
where $E$ denotes the $3\times 3$ identity matrix.

The above formulas imply
\begin{proposition}
\label{E:gauge-bracket-geometric}
The gauged nonholonomic bracket $\{\cdot , \cdot \}^B_{\mbox{\emph{\tiny nh}}}$ on $\M$,
associated to the bivector field $\pi_{\mbox{\emph{\tiny nh}}}^B$,
is given in the redundant coordinates $( g_{ij},x_k, K_l )$, $\, i,j,k,l=1,2,3$,  for $\M$ by
\begin{equation*}
\begin{split}
 \{x_i , K_l \}^B_{\mbox{\emph{\tiny nh}}}=r(Ag)_{il}, \qquad \{g_{ij} , K_l \}_{\mbox{\emph{\tiny nh}}}^B =-\varepsilon_{jl}^k\, g_{ik},
\qquad \{K_i , K_j \}^B_{\mbox{\emph{\tiny nh}}} =-\varepsilon_{ij}^l({\bf K}+mr^2(g^TA^TAg-E)\vecOm )_l \, ,
\end{split}
\end{equation*}
with all other combinations equal to zero.
\end{proposition}

\subsubsection*{Reduction  of the symmetries}

Recall that the Lie group $H$, introduced in section \ref{SSS:reduction-example}, acts on the
configuration space $Q$ and that its lift to $TQ$ leaves both the Lagrangian and the
constraints invariant. From the discussion in section \ref{SS:Reduction-general} (and the regularity
of the action) it follows that the  reduced space $\RR:=\M/H$ is equipped with a reduced bracket
$\{\cdot , \cdot \}_{\mbox{\tiny red}}$ determined by
condition \eqref{E:NHBracketReduction}  and that describes the reduced dynamics.

We are now ready to give the geometric interpretation of the bracket $\{\cdot,\cdot\}_{\mbox{\tiny Rank}j}$
introduced in section
\ref{S:Examples} for the different values of the rank of $A$.

\begin{theorem}
\label{T:red_bracket}
The reduced bracket $\{\cdot , \cdot \}_{\mbox{\tiny \emph{red}}}$ on $\RR$
is precisely the restriction of the bracket $\{\cdot,\cdot\}_{\mbox{\tiny \emph{Rank}}j}$
(defined in section \ref{S:Examples}) to the Casimir
level set $||\vecgamma||=1$, for the different values $j=0,1,2,3,$ of the rank of $A$.
\end{theorem}
\begin{proof}
Recall from section \ref{SSS:reduction-example} that the reduced space $\RR$
can be identified with $\operatorname{S}^2\times \R^3$ with redundant coordinates $(\vecgamma, {\bf K})$.
Therefore,  it makes sense to compare the two brackets  on the Casimir
level set $||\vecgamma||=1$ of the space $(\vecgamma, {\bf K})\in \R^3\times \R^3$.

 Moreover, from the expression of the projection $\rho: \M\to \RR$ given by \eqref{E:orbit-projection},
 and condition \eqref{E:NHBracketReduction}, it
 follows that the reduced bracket of the (redundant) coordinate functions $(\vecgamma, {\bf K})$
 can be computed using the formulas obtained in Proposition \ref{E:nonho-bracket-geometric} (notice
 that $\vecgamma =(g_{31}, g_{32}, g_{33})$).

 The proof is completed by considering the particular form of $A$ for the different values
 of its rank given in Definition \ref{condA}, and by writing the bracket $\{f_1 , f_2 \}_{\mbox{\tiny {red}}}$
 of arbitrary functions $f_1, f_2\in C^\infty (\RR)$ in terms of the derivatives $\frac{\partial f_i}{\partial \vecgamma}$
 and $\frac{\partial f_i}{\partial {\bf K}}$ using Leibniz rule.
\end{proof}

We now turn to the study of the reduction of the gauged nonholonomic bracket
$\{\cdot , \cdot \}^B_{\mbox{{\tiny nh}}}$. First of all notice that the 2-form $B$ that
defines the gauge transformation is written in \eqref{Ex:2formB} in terms of left invariant objects
on $\operatorname{SO}(3)$. Since the symmetry group $H$ acts by left multiplication on
the $\operatorname{SO}(3)$ factor of $Q$, it follows that $B$ is invariant under the cotangent lifted
action. Therefore, in accordance with Proposition \ref{P:reduction-gauge}, the gauged bracket
$\{\cdot , \cdot \}^B_{\mbox{{\tiny nh}}}$ drops to $\RR$ where it defines the bracket
$\{\cdot , \cdot \}_{\mbox{{\tiny red}}^B}$ that determines the dynamics. As usual, the corresponding
bivector field on $\RR$ will be denoted by $\pi_{\mbox{{\tiny red}}^B}$.

In analogy with Theorem \ref{T:red_bracket} we have
\begin{theorem}
\label{T:red_bracket_gauge}
The reduced bracket $\{\cdot , \cdot \}_{\mbox{\tiny \emph{red}}^B}$ on $\RR$
is precisely the restriction of the bracket $\{\cdot,\cdot\}'_{\mbox{\tiny \emph{Rank}}j}$
(defined in section \ref{S:Examples}) to the Casimir
level set $||\vecgamma||=1$, for the different values $j=0,1,2,3,$ of the rank of $A$.
\end{theorem}
The proof is identical to that of Theorem \ref{T:red_bracket} except that one uses the formulas
obtained in Proposition \ref{E:gauge-bracket-geometric}.

According to the discussion in Section \ref{S:Examples}, the following table summarizes the properties
of the reduced bivector fields $\pi_{\mbox{{\tiny red}}}$ and $\pi_{\mbox{{\tiny red}}^B}$ for a generalized
rolling system according to the different
values of the rank of $A$:
%

\begin{small}
\begin{center}
\begin{equation}
\label{E:Table-brackets}
  \begin{tabular}{|c|c|c|c|c| } \hline
 Rank of $A$ & 0& 1 & 2 &3 \\
    \hline
    $\pi_{\mbox{{\tiny red}}}$ &
Poisson & Conformally  Poisson
& $\begin{array}{c} \mbox{Non-integrable} \\ \mbox{characteristic } \\ \mbox{distribution} \end{array}$ &
$\begin{array}{c} \mbox{Non-integrable} \\ \mbox{characteristic } \\ \mbox{distribution} \end{array}$
 \\  \hline   $\pi_{\mbox{{\tiny red}}^B}$ &
$\begin{array}{c} \mbox{Non-integrable} \\ \mbox{characteristic } \\ \mbox{distribution} \end{array}$  &
 $\begin{array}{c} \mbox{Non-integrable} \\ \mbox{characteristic } \\ \mbox{distribution} \end{array}$  &
 Conformally  Poisson& Poisson \\
    \hline
  \end{tabular}
\end{equation}
\end{center}
\end{small}

\bigskip

\begin{remark}
\label{R:reduced_not_gauged}
Notice that the reduced bivector fields $\pi_{\mbox{{\tiny red}}}$ and $\pi_{\mbox{{\tiny red}}^B}$
are \emph{not} gauge related. Indeed, from the table above one sees that for any value of the rank
of $A$, only one of the two bivector fields
$\pi_{\mbox{{\tiny red}}}$ or $\pi_{\mbox{{\tiny red}}^B}$
has an integrable characteristic distribution.  It follows from Theorem \ref{T:GaugedRelation} that there cannot exist a gauge
transformation between their graphs.
\end{remark}

\begin{remark}
\label{R:Gauge-vs-Integrability} Recall from subsection
\ref{SS:Kinematics} that as the rank of $A$ increases, the
constraint distribution is less integrable or ``more nonholonomic".
The table \eqref{E:Table-brackets} above seems to suggest that it is
appropriate to perform a gauge transformation by the 2-form $B$
when the nonholonomic effects are more important, while the
reduction of the standard nonholonomic bracket works better for
weaker nonholonomic effects.

\end{remark}

\subsection{Hamiltonization and integrability of rigid bodies with generalized rolling constraints}
\label{SS:Hamiltonization-Integrability}

According to the notion of Hamiltonization introduced in Section \ref{Ss:ConformalFactor}
(Definition \ref{D:Hamiltonization}), and the table \eqref{E:Table-brackets}, it immediately follows
that the problem of the motion of a rigid body subject to a generalized rolling constraint
 is Hamiltonizable for any value of the rank of $A$.

If the rank of $A$ equals 0 (respectively, 3) the reduced equations are Hamiltonian
with respect to the bracket $\{ \cdot , \cdot \}_{\mbox{{\tiny red}}}$
(respectively, $\{ \cdot , \cdot \}_{\mbox{{\tiny red}}^B}$). Recall that in both cases the reduced dynamics
correspond to classical rigid body motion (with modified inertia tensor $\I+mr^2E$ if the rank of $A$
equals 3).

If the rank of $A$ equals 1 or 2, the analysis of the Hamiltonization is a bit more delicate but it also
follows directly from Definition \ref{D:Hamiltonization} and the table \eqref{E:Table-brackets}.
In the case rank $A=2$,   it follows that the reduced equations are Hamiltonian in the
 new time $\tau_2$ defined by $d\tau_2=\mbox{\small $\frac{1}{\varphi_2}$} dt$ and with respect to the
 bracket $\varphi_2 \{ \cdot , \cdot \}_{\mbox{{\tiny red}}^B}$ where
 \begin{equation}
 \varphi_2(\vecgamma) = \sqrt{1-mr^2\, (
\vecgamma \cdot (\I + mr^2E)^{-1} \vecgamma ) }.
\label{Ex:ConfFactor2}
\end{equation}
Note that $\varphi_2$ is a basic function on $\RR$ corresponding to the restriction of
 \eqref{E:rank2_conf_factor} to the level set $||\vecgamma||=1$.

Analogously, if the rank of $A$ equals 1, the reduced equations are Hamiltonian in the
new time $\tau_1$ defined by $d\tau_1=\mbox{\small $\frac{1}{\varphi_1}$} dt$ and with respect to the
 bracket $\varphi_1 \{ \cdot , \cdot \}_{\mbox{{\tiny red}}}$ where
 \begin{equation}
 \varphi_1(\vecgamma) = \sqrt{1+ mr^2\, (
\vecgamma \cdot \I^{-1} \vecgamma ) }.
\label{Ex:ConfFactor1}
\end{equation}

\subsubsection*{Integrability of the reduced equations}

In view of the Hamiltonization of the problem, the integrability of the reduced equations
of motion  \eqref{E:Motion3} can be easily established using the celebrated Arnold-Liouville Theorem
 for classical Hamiltonian systems, see e.g. \cite{Arnold}.

 Indeed, for any value of the rank of $A$, the reduced equations are
Hamiltonian on $\mathcal{R}$ (after a time reparametrization if rank $A=1,2$).
Independently of the rank of $A$, the symplectic leaves $\mathcal{O}_a$ of the foliation of $\RR$
correspond to the  level sets $C_1({\bf K}, \vecgamma)= {\bf K}\cdot \vecgamma=a$ and can be shown
to be diffeomorphic to the tangent bundle $T\operatorname{S}^2$ of the sphere (see the discussion in
chapter 14 of \cite{MarsdenRatiubook} for the coadjoint orbits  on $\se (3)^*$).

Once the value of $a$ is fixed, the reduced equations  \eqref{E:Motion3} can be seen as a two degree
of freedom classical Hamiltonian system on $\mathcal{O}_a$ (again, after a time reparametrization
if rank $A=1,2$). These equations possess two independent integrals, the Hamiltonian $\Ham_\RR$, and
$F={\bf K}\cdot{\bf K}$, whose joint level sets are compact in $\mathcal{O}_a$. It follows from the
Arnold-Liouville Theorem  that
these level sets are invariant two-tori and the dynamics are quasi-periodic on them (notice that the flow on the tori
is rectilinear but not uniform if the rank of $A$ is 1 or 2).

The Arnold-Liouville Theorem also tells us that the reduced equations are integrable by quadratures
(after the time reparametrization  if the rank of $A$ is 1 or 2).

Finally, we state without proof that the reduced equations of motion \eqref{E:Motion3} preserve the measure
$\mu (\vecgamma) \, \sigma \wedge dK_1\wedge dK_2 \wedge dK_3$ where $\sigma$ is the area form of the sphere
$\operatorname{S}^2$, and the basic density $\mu: \operatorname{S}^2\to \R$
is given by
\begin{equation*}
\mu (\vecgamma)=\begin{cases} 1  & \mbox{if rank $A=0,3$,} \\
\frac{1}{\varphi_1(\vecgamma)} & \mbox{if rank $A=1$,} \\
\frac{1}{\varphi_2(\vecgamma)} & \mbox{if rank $A=2$,}
\end{cases}
\end{equation*}
where $\varphi_1 , \varphi_2 \in C^\infty(\operatorname{S}^2)$ are
defined in \eqref{Ex:ConfFactor1} and \eqref{Ex:ConfFactor2}
respectively.

\subsection{Twisted Poisson structures for rigid bodies with generalized rolling constraints}
\label{Ss:TwistedMechanics}

In Section \ref{Sec:Twisted}, we presented twisted Poisson structures which
have been extensively studied in other contexts but not in mechanics. Now, we
will show explicitly that twisted Poisson structures appear naturally in the
study of nonholonomic systems.

\subsubsection*{Rigid body with generalized rolling constraints of rank 2}

Here we show that the bracket $\{ \cdot , \cdot \}'_{\mbox{\tiny Rank2}}$, in addition to being
conformally Poisson,  is twisted Poisson.
Note that this cannot be the case for the other bracket $\{ \cdot , \cdot \}_{\mbox{\tiny Rank2}}$
that describes the dynamics since, as shown in Section \ref{S:Examples}, its characteristic distribution
is not integrable.

Recall from the discussion in \ref{SSS:reduction-example} that $\{ \cdot , \cdot \}'_{\mbox{\tiny Rank2}}$
should be considered as a bracket on the reduced space $\RR= \operatorname{S}^2 \times \R^3$ with
redundant coordinates $(\vecgamma , {\bf K})$. The characteristic distribution of the bracket is
integrable and the leaves $\mathcal{O}_a$ of the foliation are the level sets
 $C_1(\vecgamma , {\bf K})=\vecgamma \cdot {\bf K}=a$.
By regularity and integrability of the characteristic distribution, it follows from Corollary
 \ref{C:Int_Dist_implies_twist_Poisson} that
the bracket is $\phi$-twisted. The value of the  3-form $\phi$ is given in the following,

\begin{theorem} \label{T:ChaplyginIsTwisted} The bracket $\{ \cdot , \cdot \}'_{\mbox{\tiny \emph{Rank2}}}$
defined in (\ref{E:Brackets-Rank2})
(that in particular describes the reduced dynamics of the Chaplygin sphere for the appropriate choice of $A$), is a
$\phi$-twisted Poisson bracket with $\phi = - d {\mathcal B}$ where
\begin{equation} {\mathcal B} = m r^2 (\vecOm \cdot \vecgamma)\, \sigma,
 \label{Ex:TwistedOmega}
\end{equation}
and where $\sigma$ denotes the area form of the sphere $||\vecgamma||=1$.
\end{theorem}

\begin{proof}
The idea of this proof is to show that the bracket is gauge related
to a Poisson bracket via the 2-form $-{\mathcal B}$. Thus, by
Proposition \ref{P:gauge-twist}, the bracket  is
$(-d\mathcal{B})$-twisted Poisson. More precisely, we will show that
$\{ \cdot , \cdot \}'_{\mbox{\tiny {Rank2}}}$ is
$-\mathcal{B}$-gauge related with the bracket $\{ \cdot , \cdot
\}_{\mbox{\tiny {Rank0}}}$  defined in \eqref{E:Brackets-Rank0} and
that coincides with the Lie-Poisson bracket on $\se(3)^*$.

According to Theorem \ref{T:red_bracket_gauge},  we denote the bivector field associated to the
 bracket $\{ \cdot , \cdot \}'_{\mbox{\tiny {Rank2}}}$  by  $\pi_{\mbox{{\tiny red}}^B}$.
Using  (\ref{E:Brackets-Rank2}) one gets
\begin{equation*}
(\pi_{\mbox{{\tiny red}}^B})^\sharp (d{\bf K})=\left ( {\bf K}-mr^2(\vecOm \cdot \vecgamma)\vecgamma \right )
\times \frac{\partial}{\partial {\bf K}} + \vecgamma \times \frac{\partial}{\partial \vecgamma}, \qquad
(\pi_{\mbox{{\tiny red}}^B})^\sharp (d\vecgamma)=\vecgamma \times \frac{\partial}{\partial {\bf K}}.
\end{equation*}
Next, notice that the 2-form ${\mathcal B}$ defined by
\eqref{Ex:TwistedOmega} is written in the redundant coordinates
$(\vecgamma , {\bf K})$ as
\begin{equation*}
{\mathcal B} =\frac{1}{2} m r^2 (\vecOm \cdot \vecgamma) \, \vecgamma \cdot
\left ( d\vecgamma \times d\vecgamma \right ).
\end{equation*}
The bivector field $\pi_{\mbox{{\tiny red}}^B}$ and the 2-form
$\mathcal{B}$ verify hypothesis of Proposition
\ref{P:GaugeSemibasic} and thus the gauge transformation of
$\pi_{\mbox{{\tiny red}}^B}$ associated to $\mathcal{B}$ is again a
bivector field that we will denote it by $\pi_{\mbox{{\tiny
red}}^B}^\mathcal{B}$. Relying on equation
\eqref{E:condition-gauge}, one computes
\begin{equation*}
{\bf i}_{\left (\pi_{\mbox{{\tiny red}}^B} \right )^\sharp (d{\bf K})}{\mathcal B}=
mr^2(\vecOm \cdot \vecgamma)(E-\vecgamma \vecgamma^T) \, d\vecgamma, \qquad
{\bf i}_{\left (\pi_{\mbox{{\tiny red}}^B} \right )^\sharp (d{\vecgamma})}{\mathcal B}= 0,
\end{equation*}
where, as usual, $E$ denotes the $3\times 3$ identity matrix. Using these expressions we deduce
\begin{equation*}
\left (\pi_{\mbox{{\tiny red}}^B}^\mathcal{B} \right )^\sharp(d{\bf K})=
  {\bf K} \times \frac{\partial}{\partial {\bf K}} + \vecgamma \times \frac{\partial}{\partial \vecgamma}, \qquad
(\pi_{\mbox{{\tiny red}}^B}^\mathcal{B})^\sharp (d\vecgamma)=\vecgamma \times \frac{\partial}{\partial {\bf K}},
\end{equation*}
that complete the proof.
\end{proof}

\paragraph{ The conformal factor and the 3-form $\phi$.}  In accordance with Proposition \ref{Prop:ConformallyTwisted},
since the bracket $\{ \cdot , \cdot \}'_{\mbox{\tiny Rank2}}$ is
both conformally Poisson and twisted Poisson, there is relationship
between the conformal factor $\varphi_2$ (given by
\eqref{Ex:ConfFactor2}), and the twisting 3-form $\phi$ (defined in
Theorem \ref{T:ChaplyginIsTwisted}).

We leave it to the reader to check that on the leaves
$\mathcal{O}_a$ of the foliation of $\RR$ corresponding to the
bracket $\{ \cdot , \cdot \}'_{\mbox{\tiny Rank2}}$, the 3-form
$\phi$ coincides with $\psi := \frac{1}{\varphi_2} \, d \varphi_2 \,
\wedge \Omega$, where 2-form $\Omega$ is given in the redundant
coordinates $(\vecgamma, {\bf K})$ by
\begin{equation*}
\Omega= \frac{1}{2} \left ( {\bf K} -mr^2(\vecOm \cdot \vecgamma) \vecgamma \right ) \cdot
 \left ( d\vecgamma \times d\vecgamma \right )
-\vecgamma \cdot \left ( d{\bf K} \times d\vecgamma  \right ).
\end{equation*}
This choice of $\Omega$ satisfies the conditions of   Corollary
\ref{C:DiracForm} for the graph of the bivector field
$\pi_{\mbox{{\tiny red}}^B}$ corresponding to $\{ \cdot , \cdot
\}'_{\mbox{\tiny Rank2}}$ on $T\RR \oplus T^*\RR$.

\bigskip

\subsubsection*{Rigid body with generalized rolling constraints of rank 1}

A completely analogous analysis can be performed if the rank of the
matrix $A$ equals one. This time it is the bracket $\{ \cdot , \cdot
\}_{\mbox{\tiny Rank1}}$ that is both twisted and conformally
Poisson. In analogy with Theorem \ref{T:ChaplyginIsTwisted} we have

\begin{theorem} \label{T:Rank1IsTwisted}
The bracket $\{ \cdot , \cdot \}_{\mbox{\tiny \textup{Rank1}}}$
defined in (\ref{E:Brackets-Rank1}), is a $\phi$-twisted Poisson
bracket with $\phi =  d {\mathcal B}$ with $\mathcal{B}$ given by
expression \eqref{Ex:TwistedOmega}.
\end{theorem}

%
%
%
The proof is the same to the Rank 2 case. The bracket $\{ \cdot ,
\cdot \}_{\mbox{\tiny Rank1}}$ is $\mathcal{B}$-gauge related with
the bracket $\{ \cdot , \cdot \}_{\mbox{\tiny {Rank0}}}$ defined in
\eqref{E:Brackets-Rank0} and that coincides with the Lie-Poisson
bracket on $\se(3)^*$.

\paragraph{Acknowledgments} { \small
 P.B. thanks CNPq(Brazil) and IMPA (Instituto de Matematica Pura e Aplicada, Brazil)
for supporting this project during 2009-2010 and the Centre
Interfacultaire Bernoulli, EPFL, (Switzerland) for its hospitality
during the Program {\it Advances in the Theory of Control Signals
and Systems with Physical Modeling}, where part of this work was
done.
\newline
We thank the GMC (Geometry, Mechanics and Control Network, project
MTM2009-08166-E, Spain), particularly J.C. Marrero and E. Padron, as
well as the organizers of the Young Researchers Workshops held in
Ghent and Tenerife, for facilitating our collaboration. We are also
greatly grateful to David Iglesias-Ponte, Jair Koiller, and Henrique
Bursztyn
 for useful and interesting discussions.}


\begin{thebibliography}{99}

\begin{small}

\bibitem{Arnold} V.~I. Arnold, \emph{Mathematical methods of classical mechanics.} Translated
from the Russian by K. Vogtmann and A. Weinstein. Graduate Texts in
Mathematics, {\bf 60}. Springer-Verlag, New York-Heidelberg, (1978).


\bibitem{BS93} L. Bates and J. Sniatycki, Nonholonomic reduction,  \emph{ Rep. Math. Phys.}
\textbf{32} (1993), 
99--115.


\bibitem{BKMM} A.~M. Bloch, P.~S.Krishnapasad , J.~E. Marsden
and R.~M. Murray,  Nonholonomic mechanical systems
with symmetry. {\em Arch. Rat. Mech. An.}, {\bf 136} (1996), 21--99.

\bibitem{Blochbook} A.~M. Bloch, {\em Non-holonomic mechanics and control}.
Springer Verlag, New York, (2003).


\bibitem{Mestdag-Variations} A. Bloch, O. Fernandez, and T. Mestdag,  Hamiltonization of nonholonomic systems and the
inverse problem of the calculus of variations,
{\em Rep. Math. Phys.}  {\bf 63}  (2009),  
225Ð249.

\bibitem{BorisovMamaev}  A.~V. Borisov and  I.~S. Mamaev,  Chaplygin's ball rolling
problem is Hamiltonian, \emph{ Math. Notes}, {\bf 70} (2001), 793--795.




\bibitem{BorisovMamaev2008}
 A.~V. Borisov and  I.~S. Mamaev, Conservation laws, hierarchy of dynamics and explicit
 integration of nonholonomic systems,
 \emph{Regul. Chaotic Dyn.}, {\bf 13} (2002), 
 443--490.


\bibitem{BCrainic} H. Bursztyn and  M. Crainic,  Dirac structures, momentum maps, and quasi-Poisson manifolds.
{\em The breath of Symplectic and Poisson Geometry}, 1--40. Progr.
Math.,  {\bf 232}, {\em Birkh\"auser Boston, Boston MA,} (2005).

%

\bibitem{Cantrijn99} 
 F. Cantrijn, M. de Le{\'o}n,  and D. Mart{\'i}n de
Diego,  On almost-Poisson structures in nonholonomic
mechanics, {\em Nonlinearity}, {\bf 12} (1999), 721--737.


\bibitem{chapsphere} S.~A. Chaplygin,  On a ball's rolling
on a horizontal plane, \emph{ Regul.  Chaotic Dyn.}, {\bf 7} (2002),
131--148; original paper in Mathematical Collection of the Moscow
Mathematical Society, {\bf 24} (1903), 139--168.

\bibitem{Chapligyn_reducing_multiplier} S.~A. Chaplygin,  On the theory of the motion of nonholonomic systems. The reducing-multiplier theorem.
Translated from \emph{ Matematicheski\u{i}  sbornik} (Russian)  {\bf 28} (1911) , no. 1 by A. V. Getling.  \emph{ Regul.  Chaotic Dyn.}, {\bf 13}
(2008), 
 369--376.

\bibitem{Courant} T.~J. Courant,  Dirac manifolds, \emph{Trans. Amer. Math. Soc.}, {\bf 319} (1990), 
631--661.

\bibitem{Dirac} P.~A.~M. Dirac,  Lectures on quantum mechanics. Second printing of the 1964 original. Belfer Graduate School of Science Monographs Series, 2.
\emph{Belfer Graduate School of Science, New York; produced and distributed by Academic Press, Inc., New York}, (1967).

\bibitem{EhlersKoiller} K. Ehlers, J. Koiller, R. Montgomery  and P.~M. Rios,  Nonholonomic  systems via moving frames: Cartan
equivalence and Chaplygin Hamiltonization.   {\em The breath of
Symplectic and Poisson Geometry}, 75-120. Progr.  Math.,  {\bf
232}, {\em Birkh\"auser Boston, Boston MA,} 2005.

\bibitem{FedorovJovan} 
Yu.~N. Fedorov and  B. Jovanovi{\'c},
Nonholonomic LR systems as generalized Chaplygin systems with
an invariant measure and flows on homogeneous spaces, \emph{J. Nonlinear
Sci.}, \textbf{14} (2004), 341--381.

\bibitem{FedorovKozlov} Yu.~N. Fedorov  and V.~V.  Kozlov,
 Various aspects of $n$-dimensional rigid body dynamics. {\em Dynamical systems in classical mechanics,} 141--171,
 Amer. Math. Soc. Transl. Series 2, {\bf 168}, {\em Amer. Math. Soc. Providence, RI,} 1995.


\bibitem{Fernandez} O. Fernandez, T. Mestdag and A. Bloch, A generalization of Chaplygin's reducibility Theorem
\emph{Regul. Chaotic Dyn.} {\bf 14} (2009), 
635--655.



\bibitem{Naranjo2007}
 L. Garc\'ia-Naranjo,  Reduction of almost
Poisson brackets for nonholonomic systems on Lie groups,  {\em Regul. Chaotic Dyn.},
{\bf12} (2007), 365--388.

\bibitem{Naranjo2008} L. C. Garc\'ia-Naranjo, Reduction of  almost
Poisson brackets and Hamiltonization of the Chaplygin sphere.
{\em Disc. and Cont. Dyn. Syst. Series S}, {\bf 3} (2010), 
37--60.


\bibitem{Hoch} S. Hochgerner  and L.~C. Garc\'ia-Naranjo,  $G$-Chaplygin systems with internal symmetries, truncation,
and an (almost) symplectic view of Chaplygin's ball. {\em J. Geom. Mech.} {\bf 1} (2009), 
 35--53.

\bibitem{IbLeMaMa1999} A. Ibort, M. de Le\'on, J.~C. Marrero and D. Mart\'in de Diego,   Dirac
brackets in constrained dynamics. {\em Fortschr. Phys.} {\bf 47} (1999), 
459--492.

\bibitem{JovaChap}  B. Jovanovi{\'c}, Hamiltonization and integrability of the Chaplygin sphere in $\R^n$.  \emph{J. Nonlinear Sci.}
 {\bf 20}  (2010),  
 569Ð593.


\bibitem{JotzRatiu} M. Jotz and T.~S. Ratiu,   Dirac and nonholonomic
reduction, arXiv:0806.1261,  (2008).

\bibitem{KlimcikStrobl} C. Klim\v{c}\'\i k and T. Str\"obl,  WZW-Poisson manifolds.  {\em J. Geom. Phys.} {\bf 43}  (2002),  
341--344.

\bibitem{Marle1998}
Ch.~M. Marle,   Various approaches
to conservative and nonconservative nonholonomic systems, {\em Rep. Math.
Phys.}, {\bf 42}  (1998), 211--229.


\bibitem{MarsdenRatiubook}
 J.~E. Marsden  and  T.~S.   Ratiu,
{\em Introduction to Mechanics and Symmetry. A basic exposition of classical mechanical systems.}
Second edition. Texts in Applied Mathematics, {\bf 17}.  Springer-Verlag, New York, 1999.

\bibitem{Montgomery} R. Montgomery, {\em A tour of sub-Riemannian geometries, their geodesics and applications}. Mathematical Surveys and Monographs, {\bf  91}. {\em
American Mathematical Society, Providence, RI,} 2002.

\bibitem{Ohsawa} T. Ohsawa, O. Fernandez, A. Bloch, and D. Zenkov. Nonholonomic Hamilton-Jacobi
theory via Chaplygin Hamiltonization, arXiv:1102.4361, (2011).


\bibitem{SeveraWeinstein} P. \v{S}evera and A. Weinstein,
 Poisson geometry with a 3-form background. Noncommutative geometry and string theory (Yokohama, 2001). {\em Progr. Theoret.
Phys. Suppl.}  {\bf No. 144} (2001), 145--154.

\bibitem{SchaftMaschke1994} A. J. van der Schaft and  B. M. Maschke,  On the Hamiltonian
formulation of nonholonomic mechanical systems, {\it Rep. on Math. Phys.}
{\bf 34} (1994), 225--233.

\bibitem{Veselova}  A.~P. Veselov  and L.~E. Veselova,
Integrable nonholonomic systems on Lie groups. (Russian) {\em Mat. Zametki}  {\bf 44} (1988), 
604--619, 701; {\em translation in Math. Notes} {\bf 44} 
(1989) 810--819.



\bibitem{YoshimuraMarsdenI} H. Yoshimura and J.~E., Marsden,  Dirac
structures in Lagrangian mechanics. Part I: Implicit Lagrangian
systems. {\em J.  Geom.  Phys.} {\bf  57} (2006), 
133--156

\bibitem{YoshimuraMarsdenII} H. Yoshimura and J.~E., Marsden, Dirac
Structures in Lagrangian mechanics. Part II: Variational
structures.  {\em J.  Geom.  Phys.} {\bf  57} (2006), 
 209--250

\bibitem{Weber1986} R.~W. Weber,  Hamiltonian systems with constraints
and their meaning in mechanics. {\em Arch. Rational Mech. Anal.}, {\bf 91} (1986), 
309--335.

\bibitem{Weinstein} A. Weinstein,  The modular automorphism group of a Poisson manifold, \emph{J. Geom. Phys.}  {\bf 23} (1997), 379--394.


\end{small}

\end{thebibliography}
\end{document}